\documentclass[letterpaper]{amsart}   
\usepackage{cite}
\usepackage{amsmath,amsfonts,amssymb,amsthm,amscd,mathrsfs,mathdots,bbm,color,bm}
\numberwithin{equation}{section}
\usepackage{enumerate}
\usepackage{enumitem}
\usepackage{comment}

\usepackage{graphicx}
\usepackage[linecolor=white,backgroundcolor=orange,bordercolor=red]{todonotes}

\usepackage{newtxmath}
\usepackage{newtxtext}
\usepackage{subcaption}

\usepackage{ stmaryrd }

\DeclareFontFamily{U}{mathx}{}
\DeclareFontShape{U}{mathx}{m}{n}{<-> mathx10}{}
\DeclareSymbolFont{mathx}{U}{mathx}{m}{n}
\DeclareMathAccent{\widecheck}{0}{mathx}{"71}

\def\R{{\mathbb R}}

\def\N{{\mathbb N}}

\def\AA{\mathcal{A}}
\def\BB{\mathcal{B}}
\def\FF{\mathcal{F}}
\def\SS{\mathcal{S}}
\def\HH{\mathcal{H}}
\def\KK{\mathcal{K}}

\def\Tr{\mathrm{Tr}}

\def\Op{\operatorname{Op}}

\def\sigmaPc{{\sigma}_{\Pi_N}^{\hbar}}

\newtheorem{thm}{Theorem}

\numberwithin{thm}{section}

\newtheorem{cor}[thm]{Corollary}
\newtheorem{prop}[thm]{Proposition}
\newtheorem{lem}[thm]{Lemma}
\newtheorem{definition}{Definition}

\theoremstyle{definition}

\theoremstyle{remark}
\newtheorem{rem}{Remark}

\newcommand{\be}{\begin{equation}}
\newcommand{\ee}{\end{equation}}

\newcommand{\ben}{\begin{equation*}}
\newcommand{\een}{\end{equation*}}

\def\_#1{\def\next{#1}%
 \ifx\next\risingsign\expandafter\rising\else^{\underline{#1}}\fi}
\def\risingsign{^}
\def\rising#1{^{\overline{#1}}}

\title[Truncated quantum  observables  and their semiclassical limit]{Truncated quantum  observables \\and their semiclassical limit}

\author[F.~D.~Cunden]{Fabio Deelan Cunden}
\address{Dipartimento di Matematica, Univerisit\`a degli Studi di Bari, I-70125 Bari, Italy, and
INFN, Sezione di Bari, I-70126 Bari, Italy
}
\email[Corresponding author]{fabio.cunden@uniba.it}

\author[M.~Ligab\`o]{Marilena Ligab\`o}
\address{Dipartimento di Matematica, Univerisit\`a degli Studi di Bari, I-70125 Bari, Italy}
\email{marilena.ligabo@uniba.it}

\author[M.~C.~Susca]{Maria Caterina Susca}
\address{Dipartimento di Matematica, Univerisit\`a degli Studi di Bari, I-70125 Bari, Italy}
\email{m.susca11@studenti.uniba.it}

\begin{document}

\begin{abstract}  
For quantum observables $H$ truncated on the range of orthogonal projections $\Pi_N$ of rank $N$, we study the corresponding Weyl symbol in the phase space in the semiclassical limit of vanishing Planck constant $\hbar\to0$ and large quantum number $N\to\infty$, with $\hbar N$ fixed. Under certain  assumptions, we prove the $L^2$- convergence of the Weyl symbols to a symbol truncated (hence, in general discontinuous) on the classically allowed region in phase space. As an illustration of the general theorems we analyse truncated observables for the harmonic oscillator and for a free particle in a one-dimensional box. In the latter case, we also compute the microscopic pointwise limit of the symbols near the boundary of the classically allowed region.

\end{abstract}


  \maketitle

  \thispagestyle{empty}

  \section{Introduction}
   \label{sec:Introduction}
Let $H$ be a self-adjoint operator on $L^2(\R)$, and  $\left(\Pi_N\right)_{N\geq1}$  a monotone family of orthogonal projections, $\operatorname{Ran} \Pi_1\subset\operatorname{Ran} \Pi_2\subset\operatorname{Ran} \Pi_3\subset\cdots$, with $\operatorname{rank}\Pi_N=N$ for all $N\geq1$. 
If $\operatorname{Ran} \Pi_N\subset D(H)$, it makes sense to consider the \emph{truncated quantum observables} 
$$H_N:=\Pi_NH\Pi_N,$$ 
and, under some integrability conditions, the associated Weyl symbols $\sigma_{H_N}^{\hbar}(x,p)$, where $\hbar>0$. The question addressed in this paper is the \emph{semiclassical  limit}  of  $\sigma_{H_N}^{\hbar}$, as $\hbar\to0$.

Truncated linear operators emerge in a variety of mathematical settings, and the study of their semiclassical limit is an interesting area in it own right, fueled mainly by applications in quantum theory. 
In fact, our main motivation for this study comes from the semiclassical limit of Zeno Hamiltonians and Zeno dynamics~\cite{chernoff,Facchi08,Exner07,Exner21,Facchi10,Matolcsi03,Misra}.  In the particular case $H=I$ (the identity operator), the semiclassical limit of the Weyl symbol is closely related to the scaling limits of determinantal point processes~\cite{Peres06,Bornemann16} often applied to active research topics that include eigenvalues of large random matrices, and number statistics of non-interacting fermions~\cite{Bettelheim11,Bettelheim12,Cunden18,Cunden19,Dean18,Dean19,DeBruyne21,Deift99,Deleporte21,Eisler13,Lacroix17,Nguyen22,Torquato08}. We also mention that spectral truncations (often interpreted energy cut-offs) of operator algebras have been recently studied in the development of the relation between non-commutative geometry and physics, see e.g.~\cite{Connes21,DAndrea14,Lizzi03}.

\par

The outline of the paper is as follows. In the next section we recall some some preliminary
background material. In Section~\ref{sec:main} we present a general discussion and the main results on the semiclassical limit of symbols of truncated observables in $L^2(\R)$ (Theorems~\ref{thm:main}-\ref{thm:main2}). All Theorems in Section~\ref{sec:main} can be extend to higher dimensions for truncated observables in $L^2(\R^d)$, $d\geq1$. Here we keep discussing the case $d=1$ for sake of simplicity. In Section~\ref{subs:harmonic} we illustrate the main theorems for the spectral projections of the harmonic oscillator introduced in~\cite{Cunden23} and discuss the semiclassical limit of a large class of truncated observables.
In Section~\ref{sec:box} we consider the spectral projections of a free particle in a one-dimensional box. 
The model is exactly solvable and provides a rich playground to prove the semiclassical limit of truncated observables by manipulation of explicit formulae. In addition to $L^2$-convergence we prove in this case the pointwise convergence of the symbols at microscopic scales.

\section{Preliminaries}
In this section we introduce some notation and recall definitions and basic results about Fourier transform, linear operators, Weyl symbols and Moyal product, see \cite{Ambrosio11,Figalli12,folland,Lions93,Pool, ReedSimon}.
\label{sec:preliminaries}
\paragraph{{\bf Fourier transform}} In this paper, $\FF$ and $\FF^{-1}$ denote the  Fourier transform and its inverse,
\be
\left(\FF \varphi\right)(x):=\int_{\R_p} \varphi(p)e^{-ipx}dp,\quad \left(\FF^{-1} \psi\right)(p):=\frac{1}{2\pi}\int_{\R_y} \psi(x)e^{ipx}dx.
\ee
Plancherel's theorem reads
\be
\int_{\R_p}\overline{\psi(p)}\varphi(p) dp=\frac{1}{2\pi}\int_{\R_y}\overline{\FF\psi(y)}\FF\varphi(y)dy.
\label{eq:Plancherel_f}
\ee
For functions of several variables, $\FF_j$ is the partial Fourier transform in the $j$-th variable, and  $\FF_j^{-1}$ is its inverse.
\paragraph{{\bf Linear operators and kernels}} 
For a linear  operator $T$ on $L^2(\R)$ we denote, if it exists, by $K_T(x,y)$ its integral kernel, $Tu(x)=\int_{\R_y} K_T(x,y)u(y)dy$. 
The \emph{conjugate kernel} $\widehat{K}_T:=\FF K_T\FF^{-1}$ is
\be
\widehat{K}_T\left(p,q\right)=\frac{1}{2\pi}\int_{\R_x\times\R_y} K_T(x,y)e^{-i(px-qy)}dxdy.
\ee
We have the identity
\be
\widehat{K}_T\left(\frac{1}{\hbar}p-\frac{q}{2},\frac{1}{\hbar}p+\frac{q}{2}\right)=\frac{1}{2\pi}\int_{\R_x\times\R_y}
 \hbar K_T\left(x-\frac{\hbar y}{2},x+\frac{\hbar y}{2}\right)e^{i(q x+py)}dxdy.
\ee
Moreover,  if $K_T\in L^2(\R_x\times\R_y)$ then $T$ is an Hilbert-Schmidt operator and 
$$
\|T\|_{HS}^2=\textrm{Tr}(T^\ast T)= \|K_T\|^2_{L^2(\R_x\times\R_y)}.
$$

\paragraph{{\bf Weyl symbol and Weyl quantisation}} Given $\hbar>0$, the \emph{Weyl symbol} of an operator $T$ 
is defined as
\be
\sigma_T^{\hbar}(x,p):=\int_{\R_y}\hbar K_T\left(x-\frac{\hbar y}{2},x+\frac{\hbar y}{2}\right)e^{ipy}dy.
\ee
{The above correspondence $K_T\mapsto \sigma_T$, maps $S(\R^2)$ into $S(\R_x\times\R_p)$ and extends to a map from $\mathcal{S}'(\R^2)$ into $\mathcal{S}'(\R_x\times\R_p)$. Moreover, it is unitary on $L^2(\R^2)$. (We use the standard notation $\mathcal{S}$ for the Schwarz space and $\mathcal{S}'$ for its continuous dual space i.e., the space of tempered distributions.)}   If $T$ is self-adjoint, then $K_T$ is an Hermitian kernel, i.e. $K_T(y,x)=\overline{K_T(x,y)}$, and the symbol $\sigma_T^{\hbar}$ is real-valued.

In terms of Fourier transform, the Weyl symbol $\sigma_T^{\hbar}$ is 
\be
\left(\FF_2 \sigma_T^{\hbar}\right)(x,y)=\left(2\pi\hbar\right) K_T\left(x-\frac{\hbar y}{2},x+\frac{\hbar y}{2}\right).
\ee
The relation with the conjugate kernel $\widehat{K}_T$ is
\be
\sigma_T^{\hbar}(x,p)=\frac{1}{2\pi}\int_{\R_q}
\widehat{K}_T\left(\frac{1}{\hbar}p-\frac{q}{2},\frac{1}{\hbar}p+\frac{q}{2}\right)e^{i q x}dq,
\ee
or equivalently
\be
\left(\FF_1^{-1} \sigma_T^{\hbar}\right)(q,p)=
\widehat{K}_T\left(\frac{1}{\hbar}p-\frac{q}{2},\frac{1}{\hbar}p+\frac{q}{2}\right).
\ee
From Plancherel's theorem, if $T$ is Hilbert-Schmidt, then
\be
\label{eq:useful_identity_norm}
\|\sigma^{\hbar}_T\|^2_{L^2(\R_x\times\R_p)}=2\pi\hbar \|K_T\|^2_{L^2(\R_x\times\R_y)}=2\pi\hbar \|\widehat{K}_T\|^2_{L^2(\R_q\times\R_p)}.
\ee
If $f=f(x,p)$ is a real-valued function, its Weyl quantisation $\Op^{\hbar}(f)$ is an operator acting (on a suitable subspace of  $L^2(\R)$) as
\be
\label{eq:Weyl_quantisation}  
\Op^{\hbar}(f)\psi(x)=\int_{\R_y\times\R_p} f\left(\frac{x+y}{2},p\right)e^{i\frac{p}{\hbar}(x-y)}\psi(y)\frac{dydp}{2\pi\hbar}.
\ee

\paragraph{{\bf Moyal product}} 
Given two linear operators $T_1$ and $T_2$ on $L^2(\R)$ such that $\operatorname{Ran}(T_2)\subset D(T_1)$, with Weyl symbols $\sigma_{T_1}^\hbar$ and $\sigma_{T_2}^\hbar$ respectively, the \emph{Moyal product} of $\sigma_{T_1}^\hbar$ and $\sigma_{T_2}^\hbar$, denoted $\sigma^{\hbar}_{T_1} \, \sharp \, \sigma^{\hbar}_{T_2}$ is the function defined by the equation
\be
\sigma^{\hbar}_{T_1} \, \sharp \, \sigma^{\hbar}_{T_2}(x,p)=\sigma^{\hbar}_{T_1T_2}(x,p).
\ee
Explicit formulae are:
\begin{align}
\label{eq:Moyal1}
\sigma^{\hbar}_{T_1} \, \sharp \, \sigma^{\hbar}_{T_2}(x,p)&=\int_{\R^4} e^{\frac{2i}{\hbar}[(x-x_1)(p-p_2)-(x-x_2)(p-p_1)]}\prod_{i=1,2}\sigma^{\hbar}_{T_i}(x_i,p_i)\frac{dx_i dp_i}{\pi \hbar}\\
\label{eq:Moyal2}
&=\int_{\R^4} e^{-i(k_1x_2-k_2x_1)}\prod_{i=1,2}\sigma^{\hbar}_{T_i}\left(x-x_i,p-\frac{\hbar k_i}{2}\right)\frac{dx_i dk_i}{2\pi }\\
\label{eq:Moyal3}
&=\int_{\R^4} e^{-i(p_1y_2-p_2y_1)}\prod_{i=1,2}\sigma^{\hbar}_{T_i}\left(x-\frac{\hbar y_i}{2},p-p_i\right)\frac{dy_i dp_i}{2\pi }.
\end{align}
\paragraph{{\bf Dirichlet and sine kernels}} 
The \emph{Dirichlet kernel} is defined as
\be
\label{eq:Dirichletkernel}
D_N(x):=\sum_{|k|\leq N}e^{i k x}=
\left(1+2\sum_{k=1}^N\cos(kx)\right)=
\begin{cases}
\dfrac{\sin\left((2N+1)\frac{x}{2}\right)}{\sin\left(\frac{x}{2}\right)}&\text{if $x\neq0$}\\
2N+1&\text{if $x=0$}
\end{cases}.
\ee
for all nonnegative integers $N$. The \emph{sine kernel} is defined as
\be
\label{eq:sinekernel}
S(x):=\int_{-\frac{1}{2}}^{\frac{1}{2}}e^{ikx}dk=\frac{\sin\frac{x}{2}}{\frac{x}{2}}.
\ee
(It is the Fourier transform of the indicator function $\chi_{\left[-\frac{1}{2},\frac{1}{2}\right]}$.)
\par

\section{General discussion and main results}
\label{sec:main}
To get an idea of what one can expect as a limit for $\sigma_{H_N}^{\hbar}$, we consider the simplest setting where  $H=I$ is the identity operator, and so $H_N=\Pi_N$. 
The orthogonal projections $\Pi_N$ can be represented as
\be
\Pi_{N}=\sum_{k=1}^N |u_k\rangle \langle u_k|,
\ee
for some orthonormal basis $(u_k)_{k\geq1}$ of an infinite dimensional subspace of $L^2(\R)$ (hereafter we use the Dirac notation $|u_k\rangle \langle u_k|$ to denote the projection onto $\textrm{span}\{u_k\}$). 
A first remark should be made about the symbol  of $\Pi_{N}$: 
\begin{equation}
\sigma_{\Pi_N}^{\hbar}(x,p)= \sum_{k=1}^N \sigma_{|u_k\rangle \langle u_k|}^{\hbar}(x,p) = \hbar \sum_{k=1}^N \int_{\R_y} \overline{u_k\left( x-\frac{\hbar y}{2}\right)}u_k\left( x+\frac{\hbar y}{2}\right) e^{-ipy} \, dy .
\end{equation}
{Hence, $\|\sigma_{\Pi_N}^{\hbar}\|_{L^{\infty}(\R_x\times\R_p)}\leq \hbar N$, from  the Cauchy-Schwarz inequality and the condition $\|u_k\|_{L^2(\R)}=1$ for all $k\geq1$. Therefore~\cite[Prop. (1.92)]{folland}, $\sigma_{\Pi_N}^{\hbar}\in L^2(\R_x\times\R_p)\cap  C_0(\R_x\times\R_p)$. Moreover, $\|\sigma_{\Pi_N}^{\hbar}\|_{L^2(\R_x \times \R_p)}=\sqrt{2\pi\hbar N}$.
Then, for all $\varphi\in L^2(\R_x\times\R_p)$, we have the basic estimate}
\begin{align}
\label{eq:basic_est}
\left|\int_{\R_x\times\R_p}\sigma_{\Pi_N}^{\hbar}(x,p)\varphi(x,p)dxdp\right|\leq \sqrt{2\pi \hbar N}\|\varphi\|_{L^2}.
\end{align} 
Therefore, $\sigma_{\Pi_N}^{\hbar}$ has a trivial weak limit, as $\hbar\to0$. In fact, the basic estimate~\eqref{eq:basic_est} suggests that something interesting can emerge in the simultaneous limit $\hbar\to0$, $N\to\infty$ with the product $\hbar N=\mu>0$ fixed. 

Indeed, when $\hbar N=\mu$, the family $(\sigma_{\Pi_N}^{\hbar})_{N\geq1}$ is bounded in $L^2(\R_x\times\R_p)$ and therefore has a limit $\sigma\in L^2(\R_x\times\R_p)$, upon extraction of a subsequence. If the symbols $(\sigma_{\Pi_N}^{\hbar})_{N\geq1}$ were regular, say $C^1(\R_x\times\R_p)$, then 
\be
\sigma_{\Pi_N}^{\hbar}=\sigma_{\Pi_N^2}^{\hbar}=\sigma_{\Pi_N}^{\hbar}\, \sharp\,\sigma_{\Pi_N}^{\hbar}=\left(\sigma_{\Pi_N}^{\hbar}\right)^2+O(\hbar)
\ee
as $\hbar\to0$, see \cite{folland}. Hence, the limit must satisfy $\sigma=\sigma^2$: the limit symbol can only take value $0$ or $1$, excluding Lebesgue negligible sets. We conclude that if the limit is not zero, then it must be the \emph{characteristic function} $\sigma=\chi_{\Omega}$ of a \emph{bounded} region $\Omega\in\R_x\times\R_p$ in the phase space. This informal discussion can be made precise when the $\Pi_N$'s are the spectral projections of Schr\"odinger operators: in that case we expect the symbols $\sigma_{\Pi_N}^{\hbar}$ to concentrate on the corresponding classically allowed region $\Omega$. (See also~\cite{Hernandez-Duenas15}.)

Consider the Schr\"odinger operator $A=\frac{1}{2}\hat{p}^2+V(\hat{x})$, where $\hat{p}=\frac{\hbar}{i}\frac{d}{dx}$ and $V(\hat{x})$ the multiplication operator by $V(x)$. We work in the convenient setting where $V\colon\R\to\R$ is $C^{\infty}(\R)$, bounded from below, and confining, $|V(x)|\to+\infty$, if $|x|\to+\infty$. Then, $A$ is essentially self-adjoint on  $D(A)=H^2(\R)\cap\{u\in L^2(\R)\colon Vu\in L^2(\R)\}$ and has compact
resolvent: it admits a non-decreasing sequence $(E_k)_{k\geq1}$ of eigenvalues tending to $+\infty$, and the associated eigenvectors $(u_k)_{k\geq1}$  form a basis in $L^2(\R)$.
Note that the eigenvalues and eigenfunctions of $A$ depend on $\hbar$,  in general. By an appropriate shift of the potential $V(x)$ we can assume without loss of generality that $A$ is positive.

The Schr\"odinger operator $A$ is the Weyl quantisation $A=\Op^{\hbar}(a)$ of the \emph{classical Hamiltonian function} 
\be
a(x,p)=\frac{1}{2}{p}^2+V({x}),
\ee
see~\eqref{eq:Weyl_quantisation}.
Let $\Pi_N=\chi_{(-\infty,E_N]}(A)$ be spectral projections of $A$. By the integrated Weyl law (see for instance ~\cite[Thm. 6.8]{Zworski12}), 
\be
\lim_{\substack{N\to\infty, \hbar\to0\\\hbar N=\mu}}E_N= E
\ee
with $E=E(\mu)$ being given by the solution of
\be
\mu=\frac{1}{\pi}\int_{\Omega}dxdp=\frac{1}{\pi}\int_{\R}\sqrt{2\left(E-V(x)\right)_+}dx,
\ee 
where 
\be
\Omega=\{(x,p)\in\R_x\times\R_p\colon a(x,p)\leq E\}
\ee 
is the \emph{classically allowed region} (we omit the dependence of $\Omega$ on $E$ or $\mu$). Then, in the simultaneous limit of small Planck's constant $\hbar$ and large quantum number $N$, one expects that $\sigmaPc$ converges to $\chi_{\Omega}$. Here is a precise statement.
\begin{thm} 
\label{thm:main}
Set $\mu>0$. Then, under the above hypotheses on  $V(x)$, 

 \begin{align}
\label{eq:asymp_sigma_P_<N}
\lim_{\substack{N\to\infty, \hbar\to0\\\hbar N=\mu}}\|\sigmaPc -\chi_{\Omega}\|_{L^2(\R_x\times\R_p)}=0.
\end{align}
\end{thm}
The proof of Theorem~\ref{thm:main}  is postponed to Section \ref{sebsect:proof}.

\begin{rem} \label{rem:1} 
{Theorem~\ref{thm:main} does \emph{not} imply pointwise convergence of the symbols (namely, that $(\sigmaPc-\chi_{\Omega})(x,p)\to 0$ for all $(x,p)\in\R_x\times\R_p$), in general.} Suppose for instance that $V(x)$ is even. Then, the eigenfunctions of $A$ have defined parity $u_k(-x)=(-1)^{k-1}u_k(x)$, $k\geq1$. It is then easy to verify that $\sigmaPc(0,0)
=1+(-1)^{N+1}$. {See Fig.~\ref{fig:symbol}.} The strong $L^2$-convergence only implies the existence of a {subsequence} $\sigma^{\hbar}_{\Pi_{N_k}}$ that converges {almost everywhere} to $\chi_{\Omega}$.
\end{rem}
\begin{rem} The fact that the `macroscopic' limit of the symbols is a characteristic function is a universality result. It is the translation (in the phase space) of the universality of the \emph{sine kernel} for the integral kernel of spectral projections~\cite{Deleporte21}. The universality of the sine kernel is well-known for the \emph{Christoffel-Darboux kernels} of families of orthogonal polynomial on the real line~\cite{Simon08}. Theorem~\ref{thm:main} can be extend to dimension $d\geq1$ (with the sine kernel replaced by more general Bessel type kernels).
\end{rem}
\begin{rem} Note that the limit symbol $\chi_{\Omega}$ is a discontinuous function in $\R_x\times\R_p$. A natural question is what happens at the boundary $\partial\Omega$ in a `microscopic limit', i.e. if we consider the limit of the symbols $\sigma_{\Pi_N}^{\hbar}(x_{\hbar},y_{\hbar})$ with  sequences $(x_{\hbar},y_{\hbar})$ converging to the boundary $\partial\Omega$ as $\hbar\to0$. An analogue of Theorem~\ref{thm:main} can be stated in terms of rescaled Fourier transforms of the symbols. Suppose that $x_0$ is a point of inversion for the classical motion, i.e. $V(x_0)=E$. If $V'(x_0)\neq0$, then for all compact sets $W\Subset\R$, there exist constants $\hbar_0>0$ and $C>0$ such that
\begin{align}
\label{eq:estimate_conv_edge}
&\sup_{ y\in W}\left|\hbar^{-\frac{1}{3}}\FF_2\sigmaPc\left(x_0,\hbar^{-\frac{1}{3}}y\right)-\FF_2{\chi_{\Omega}^{(N)}}\left(y\right)\right|\leq C\hbar^{\frac{1}{3}}
\end{align}
for all $0<\hbar<\hbar_0$, where $\FF_2{\chi_{\Omega}^{(N)}}$ is the \emph{integrated Airy function}, see~\cite{Cunden23}. This universal result (adapted  from~\cite{Deleporte21}) should be compared to the key estimate~\eqref{eq:estimate_conv} underlying Theorem~\ref{thm:main}.  
We stress that the microscopic limit of the symbol of the projections at the boundary of the classically allowed region depends on the behaviour of the potential $V(x)$ at the points $x_0$ of inversion of classical motion. When $V$ is  regular with non-zero derivative at $x_0$,  it is  natural to expect the \emph{Airy scaling} and asymptotic symbols expressed in terms of Airy functions. Different behaviours of the potential $V(x)$ give rise to different scaling limits. See Theorem~\ref{thm:pointwise} in Section~\ref{sub:pointwise} for an explicit calculation of the microscopic limit of the symbols for a free particle confined in a box ($V(x)$ is formally infinite outside an interval). These  models with `hard wall potentials' were studied in detail in~\cite{DeBruyne21}.   
\end{rem}

If we now consider the truncated operators $H_N=\Pi_NH\Pi_N$, one expects the symbol  $\sigma_{H_N}^{\hbar}$  to converge to the restriction of the \emph{principal symbol} of $H$ to the classically allowed region, namely $\sigma_H\chi_\Omega$, where, rather informally, $\sigma_H=\lim_{\hbar\to0}\sigma_H^{\hbar}$. Let us consider $H=\Op^{\hbar}(f)$  the Weyl quantisation of  a real-valued function $f(x,p)$ on the phase space, with $\operatorname{Ran}(\Pi_N)\subset D(H)$. This choice simplifies the problem since the symbol of $\Op^{\hbar}(f)$ is just
\be
\sigma_{\Op^{\hbar}(f)}^{\hbar}(x,p)=f(x,p),
\ee
and thus coincides with its principal symbol (independent of $\hbar$). 
As special cases the reader can think of $H=f(\hat{x})$ being the multiplication of operator by $f(x)$ in the position representation or $H=f(\hat{p})$ the multiplication operator by $f(p)$ in the momentum representation. 

Note that in general $\Op^{\hbar}(f)$ and $\Pi_N$ do \emph{not} commute, and therefore 
$$
\sigma_{\Pi_N \Op^{\hbar}(f)\Pi_N}^{\hbar}\neq \sigmaPc\sigma_{\Op^{\hbar}(f)}^{\hbar}\sigmaPc.
$$
Nevertheless for $\hbar\to0$ one expects the algebra of operators to reduce to a commutative algebra, so that
\begin{align*}
\sigma_{\Pi_N\Op^{\hbar}(f)\Pi_N}^{\hbar}(x,p)&=\sigma_{\Op^{\hbar}(f)\Pi_N\Pi_N}^{\hbar}(x,p)+\sigma_{[\Pi_N,\Op^{\hbar}(f)]\Pi_N}^{\hbar}(x,p)\\
&\sim\sigma_{\Op^{\hbar}(f)\Pi_N}^{\hbar}(x,p)\\
&=\sigma_{\Op^{\hbar}(f)}^{\hbar}\, \sharp \,\sigma_{\Pi_N}^{\hbar}(x,p)\\ 
&=f\, \sharp \,\sigma_{\Pi_N}^{\hbar}(x,p)\\ 
&\sim f(x,p)\sigma_{\Pi_N}^{\hbar}(x,p).
\end{align*}
The first approximation symbol `$\sim$' amounts to say that  $\Op^{\hbar}(f)$ and $\Pi_N$ asymptotically commute on the range of $\Pi_N$. The second `$\sim$'  says that in the semiclassical limit, the noncommutative Moyal product of $f$ and the symbols of   $\sigma_{\Pi_N}^{\hbar}$ reduces to the commutative pointwise product.
\par

\par

The next Theorems provide sufficient conditions for those two `approximations' to hold (recall that here $\hbar\to0$ and $N\to\infty$ simultaneously) thus implying that, if the symbols of $\Pi_N$ have a weak limit, then do so the symbols of the truncated observables $\Pi_N^{}\Op^{\hbar}(f)\Pi_N$. 

In the following, $\left(\Pi_N\right)_{N\geq1}$ is a monotone family of orthogonal projections $\operatorname{Ran} \Pi_1\subset\operatorname{Ran} \Pi_2\subset\operatorname{Ran} \Pi_3\subset\cdots$, with $\operatorname{rank}\Pi_N=N$ for all $N\geq1$.  Denote by $\Pi_N^{\perp}:=1-\Pi_N$ the orthogonal projection onto $\operatorname{Ran}(\Pi_N)^{\perp}$. We assume that $H$ is a (densely defined) self-adjoint operator on $L^2(\R)$ with $\operatorname{Ran}(\Pi_N)\subset D(H)$,  for all $N\geq1$. 
\begin{definition} Let $\mu>0$, and set $\hbar=\mu/N$.  Define the following  Conditions on $\left(\Pi_N\right)_{N\geq1}$ and $H$:
\label{def:conditions}
\begin{enumerate}[label={(C\arabic*)},ref=C\arabic*]\addtocounter{enumi}{-1}
\item \label{item:cond1_gen} The sequence $\left(\sigmaPc\right)_{N\geq1}$ weakly converges to $\sigma\in L^2(\R_x\times\R_p)$, as $N\to\infty$, with $\hbar N=\mu$;
\item\label{item:cond2_gen} The sequence $\left(\sigma_{\Pi_N H\Pi_N}^{\hbar}\right)_{N\geq1}$ is  bounded in $L^2(\R_x\times\R_p)$; 
\item \label{item:cond3_gen} $\lim\limits_{\substack{N\to\infty, \hbar\to0\\\hbar N=\mu}}\|\sigma_{\Pi_N^{\perp}H \Pi_N}^{\hbar}\|_{L^2(\R_x\times\R_p)}=0$.
\end{enumerate}
\end{definition}

Theorem~\ref{thm:main} provides a large class of orthogonal projections satisfying Condition~\eqref{item:cond1_gen}. 
Clearly, all bounded operators $H$ satisfy~\eqref{item:cond2_gen}, and Condition~\eqref{item:cond3_gen} is trivially met if $H$ and $\Pi_N$ commute. However the class of operators obeying Conditions~\eqref{item:cond2_gen}-\eqref{item:cond3_gen} is much larger. See, e.g., the guiding example in Section~\ref{subs:harmonic}. Typically, condition~\eqref{item:cond3_gen} is  the harder to check. 

\begin{thm}
\label{thm:main2} Let $\mu>0$, and set $\hbar=\mu/N$. Let $f_1$ and $f_2$ be continuous real-valued function, and let 
$$H=\Op^{\hbar}(f),\quad f(x,p)=f_1(x)+f_2(p)$$ 
be defined by
\be
H\psi(x)=f_1(x)\psi(x)+\int_{\R_p}\left(\int_{\R_y}f_2(p)e^{i\frac{p}{\hbar}(x-y)}\psi(y)dy\right)\frac{dp}{2\pi\hbar}.
\ee 
If Conditions~\eqref{item:cond1_gen}-\eqref{item:cond2_gen}-\eqref{item:cond3_gen} hold, 
then $\sigma_{\Pi_N^{}\Op^{\hbar}(f)\Pi_N}^{\hbar}$ weakly converges to $f\sigma\in L^2(\R_x\times\R_p)$:
\be
\lim_{\substack{N\to\infty, \hbar\to0\\\hbar N=\mu}}\langle\sigma_{\Pi_N^{}\Op^{\hbar}(f)\Pi_N}^{\hbar}-f\sigma,\varphi\rangle_{L^2(\R_x\times\R_p)}=0,\quad\text{for all $\varphi\in L^2(\R_x\times\R_p)$}.
\ee
If one replaces~\eqref{item:cond2_gen} by
\begin{enumerate}[label={(C\arabic*')},ref=C\arabic*']
\item \label{item:condp2} $\|\sigma_{\Pi_N\Op^{\hbar}(f)\Pi_N}^{\hbar}\|_{L^2}\to \|f\sigma\|_{L^2}$,
\end{enumerate}
then, $\sigma_{\Pi_N^{}\Op^{\hbar}(f)\Pi_N}^{\hbar}$  strongly converges  to $f\sigma$ in $L^2(\R_x\times\R_p)$:
\be
\lim\limits_{\substack{N\to\infty, \hbar\to0\\\hbar N=\mu}}\|\sigma_{\Pi_N^{}\Op^{\hbar}(f)\Pi_N}^{\hbar}-f\sigma\|_{L^2(\R_x\times\R_p)}=0.
\ee
\end{thm}

We can extend the theorem to a class of operators $H=\mathrm{Op}^{\hbar}\left(f\right)$ that `mix' $\hat{x}$ and $\hat{p}$, under some regularity assumptions.
\par

\begin{thm}
\label{thm:main3}
 Let $\mu>0$ and set $\hbar=\mu/N$.  Let $f\in C^{\infty}(\R_x\times\R_p)$ be a real-valued polynomial in $p$ with polynomially bounded coefficients in $x$. 
 Consider
$$H=\mathrm{Op}^{\hbar}\left(f\right)$$ 
the Weyl quantisation of $f$ defined by~\eqref{eq:Weyl_quantisation}  
on its maximal domain $D(H)$. (Namely, $H$ is a differential operator with polynomially bounded coefficients.) 

 If Conditions~\eqref{item:cond1_gen}-\eqref{item:cond2_gen}-\eqref{item:cond3_gen} hold, then
$\sigma_{\Pi_N^{}\Op^{\hbar}(f)\Pi_N}^{\hbar}$ weakly converges to $f\sigma\in L^2(\R_x\times\R_p)$.

If one replaces~\eqref{item:cond2_gen} by
\begin{enumerate}[label={(C\arabic*')},ref=C\arabic*']
\item \label{item:condp2d} $\|\sigma_{\Pi_N\Op^{\hbar}(f)\Pi_N}^{\hbar}\|_{L^2(R_x\times\R_p)}\to \|f\sigma\|_{L^2(R_x\times\R_p)}$,
\end{enumerate}
then, $\sigma_{\Pi_N^{}\Op^{\hbar}(f)\Pi_N}^{\hbar}$  strongly converges  to $f\sigma$ in $L^2(\R_x\times\R_p)$.

The same conclusions hold if $f=f(x,p)$ is a real-valued polynomial in $x$ with polynomially bounded coefficients in $p$ and, in particular, if $f=f(x,p)$ is a real-valued polynomial in both $x$ and $p$.
\end{thm}
The proofs of Theorem~\ref{thm:main2} and Theorem~\ref{thm:main3}  are postponed to Section \ref{sebsect:proof}.

\subsection{Guiding example: the quantum harmonic oscillator}
\label{subs:harmonic}
We apply here the main theorems to a specific model that can be thought of as a guiding example. The model fits in the hypotheses of Theorems~\ref{thm:main}-~\ref{thm:main2}, and was studied  in~\cite{Cunden23} motivated by an experimental proposal of a quantum Zeno dynamics in QED cavity by Raimond et al.~\cite{Raimond10,Raimond12}.

\par

Let $(u_k)_{k\geq1}$ be the Hermite wavefunctions defined as the normalised eigenfunctions of the quantum harmonic oscillator $A=\frac{1}{2}\left(\hat{p}^2+\hat{x}^2\right)$:
$$
u_k(x)=\frac{1}{\sqrt{2^{k-1} (k-1)! }}\left( \frac{1}{\pi \hbar}\right)^{1/4}e^{-\frac{x^2}{2 \hbar}}H_{k-1}\left( \frac{x}{ \sqrt{\hbar}}\right), \quad k\geq 1,
$$
where $(H_j)_{j \in \N}$ are the Hermite polynomials
$$
H_j(y)=(-1)^j e^{y^2} \frac{d^j}{d y^j}(e^{-y^2}), \quad j \in \N,
$$ 
and the corresponding eigenvalues are $E_k=\hbar(k-1/2)$, $k \geq 1$.  The classical hamiltonian of the harmonic oscillator is $a(x,p)=\frac{1}{2}\left(p^2+x^2\right)$.  
Throughout this section we consider the nested sequence of orthogonal projections
\begin{equation}
\label{eq:proj_hermite}
\Pi_{N}:=\chi_{(-\infty,E_N]}(A)=\sum_{k=1}^N |u_k\rangle \langle u_k|,\quad N\geq1.
\end{equation}

If $\mu >0$, then in the  limit $\hbar\to0,N\to\infty$, with $\hbar N=\mu$, we have $E_N\to \mu$, therefore, the classically allowed region $\Omega=D$ is the disk of radius $\sqrt{2\mu}$,
\be
\label{eq:disk}
D=\left\{(x,p)\in\R_x\times\R_p\colon a(x,p)\leq \mu\right\}=\left\{(x,p)\in\R_x\times\R_p\colon x^2+p^2\leq 2\mu\right\}.
\ee

 As an example of truncated observables let us consider the truncation of $H=\Op^{\hbar}(f)$, with $f(x,p)=(ax+bp)^n$, with $a,b\in\R$ and $n$ nonnegative integer, on the range of $\Pi_N$.
  \begin{prop}\label{thm:convergenceAprime_Hermite2} 
Set $\mu>0$. 
  Then, 
 \begin{align}
\label{eq:asymp_sigma_H_N_ho2}
\lim_{\substack{N\to\infty, \hbar\to0\\\hbar N=\mu}}\|\sigma^{\hbar}_{\Pi_N\Op^{\hbar}(f)\Pi_N}-f\chi_{D}\|_{L^2(\R_x\times\R_p)}&=0,
\end{align}
where $\chi_{D}$ is the characteristic function of the disk $D$ in~\eqref{eq:disk}.
\end{prop}
 \begin{proof}
Since $f$ is a polynomial in $x$ and $p$,
to prove the limit~\eqref{eq:asymp_sigma_H_N_ho2} it is enough to check the conditions of Theorem~\ref{thm:main3}.  For the harmonic oscillator the confining potential is $V(x)=x^2/2$; hence from Theorem~\ref{thm:main} we have that the symbol $\sigma^{\hbar}_{\Pi_N}$ strongly converges in $L^2$ to the characteristic function $\chi_{D}$, 
\be
\label{eq:asymp_sigma_P_<N_ho_2}
\lim_{\substack{N\to\infty, \hbar\to0\\\hbar N=\mu}}\|\sigma^{\hbar}_{\Pi_N}-\chi_{D}\|_{L^2(\R_x\times\R_p)}=0.
\ee
A standard calculation from~\eqref{eq:Weyl_quantisation}  shows that the Weyl quantisation of $f$ is    
 $$\Op^{\hbar}(f)=(a\hat{x}+b\hat{p})^n.$$
Using 
\be
\label{eq:useful_Hemite}
\begin{aligned}
\hat{x}u_{k}&=\sqrt{\frac{\hbar}{2}}\left[\sqrt{k+1}u_{k+1} +\sqrt{k}u_{k-1} \right],
\\
\hat{p}u_{k}&=i\sqrt{\frac{\hbar}{2}}\left[\sqrt{k+1}u_{k+1} -\sqrt{k}u_{k-1} \right],
\end{aligned}
\ee
we can give a combinatorial formula for the matrix elements of $\Op^{\hbar}(f)$ in the basis $(u_k)_{k\geq1}$.
Let $k, l$ be two nonnegative integers. A $n$-steps path from $k$ to $l$ is a $n$-tuple
$$ 
\gamma=\left((j_0,j_1),(j_1,j_2),\ldots, (j_{n-1},j_n)\right) \in \underbrace{\N^2 \times \dots \times \N^2}_{ \textrm{$n$ times}},
$$ 
such that $j_0=k$, $j_n=l$, and $|j_{m}-j_{m+1}|=1$ for all $m=0,\ldots,n-1$. To a $n$-steps path  $\gamma$  we assign the weight 
$$
w(\gamma):=\prod_{m=0}^{n-1}\left(a+(j_{m+1}- j_{m})ib\right)\sqrt{\max(j_{m+1}, j_{m})}.
$$
Let $\Gamma_n(k,\ell)$ be the set of $n$-steps paths from $k$ to $\ell$. Note that $\Gamma_n(k,\ell)=\emptyset$ whenever $|k-\ell|>n$.
 Then, from~\eqref{eq:useful_Hemite} 
we can write   
\be \label{eqn:matrixelementi}
\langle{u_{\ell},\Op^{\hbar}(f)u_{k}\rangle_{L^{2}(\R)}}=
\begin{cases}
\displaystyle\left(\frac{\hbar}{2}\right)^{\frac{n}{2}}\sum_{\gamma\in\Gamma_n(k,\ell)}w(\gamma)&\text{if $|k-\ell|\leq n$}\\
0&\text{if $|k-\ell|> n$}
\end{cases}
\ee
We have that for $k\leq\ell$, with $\ell-k\leq n$,
\be
\sum_{\gamma\in\Gamma_n(k,\ell)}w(\gamma)=(a+ib)^{\ell-k}(a-ib)^{n-(\ell-k)}\#\Gamma_n(k,\ell)k^\frac{n}{2}\left[1+o(1)\right],
\ee 
as $k\to\infty$, where the total number of $n$-steps paths from $k$ to $\ell$ is
\be
\label{eq:number_paths}
\#\Gamma_n(k,\ell)=\binom{n}{\ell-k}
\ee
Hence, from~\eqref{eqn:matrixelementi} we see that $\Op^{\hbar}(f)$ is a finite-band operator in the basis $(u_k)_{k\geq1}$ of the eigenfunctions of $A$. This implies that $\Pi_N^{\perp}\Op^{\hbar}(f)\Pi_N$ has rank $n$ and 
\be
\begin{aligned}
\|\sigma_{\Pi_N^{\perp}\Op^{\hbar}(f)\Pi_N}^{\hbar}\|^2_{L^2(\R_x\times\R_p)}&=2\pi\hbar \left(\frac{\hbar}{2}\right)^{n}\sum_{k\leq N}\sum_{\ell=N+1}^{N+k}\left|\sum_{\gamma\in\Gamma_n(k,\ell)}w(\gamma)\right|^2\\
& \leq c (a^2+b^2)^{n} \hbar^{n+1} N^n,
\end{aligned}
\ee
where $c$ is a constant dependent on $n$ but not on $N$.
Hence Condition~\eqref{item:cond3_gen}  holds true. We now show that Condition~\eqref{item:condp2} is also met. 
We have
\begin{align*}
\|\sigma_{\Pi_N\Op^{\hbar}(f)\Pi_N}^{\hbar}\|^2_{L^2(\R_x\times\R_p)}&=2\pi\frac{\hbar^{n+1}}{2^n}\sum_{k,\ell\leq N}\left|\sum_{\gamma\in\Gamma_n(k,\ell)}w(\gamma)\right|^2\\
&=2\pi\frac{\hbar^{n+1}}{2^n}(a^2+b^2)^{n}\sum_{k\leq N}k^n\sum_{j=0}^n\binom{n}{j}^2[1+o(1)], \quad\text{as $N\to\infty$}.
\end{align*}
 We have 
$$\sum_{k\leq N}k^n=\frac{N^{n+1}}{n+1}[1+o(1)], \quad\text{as $N\to\infty$},$$ and  $$\sum_{j=0}^n\binom{n}{j}^2=\binom{2n}{n},$$
by Vandermonde's identity.
Therefore,
\be
\lim\limits_{\substack{N\to\infty, \hbar\to0\\\hbar N=\mu}}\|\sigma_{\Pi_N\Op^{\hbar}(f)\Pi_N}^{\hbar}\|^2_{L^2(\R_x\times\R_p)}=2 \pi   \mu ^{n+1}  \left(\frac{a^2+b^2}{2}\right)^nC_n, 
\ee
where $C_n=\frac{1}{n+1}\binom{2 n}{n}$ is the $n$-th Catalan number.
The limit is indeed equal to 
$$
\|f\sigma_D\|^2_{L^2(\R_x\times\R_p)}=\int_D(ax+bp)^{2n}dxdp=\int_0^{\sqrt{2\mu}}r^{2n+1}dr \int_0^{2\pi} (a\cos\theta+b\sin\theta)^{2n}d\theta,$$ 
see the Lemma below.
\end{proof}
\begin{lem}
\label{lem:integral}
\be
I=\int_0^{2\pi}(a\cos\theta+b\sin\theta)^{2n}d\theta=2 \pi\frac{ \left(a^2+b^2\right)^n}{2^{2 n}} \binom{2 n}{n}.
\ee
\end{lem}
\begin{proof}
Setting $z=e^{i\theta}$ we get
\begin{align*}
I&=\oint_{|z|=1}\left(\frac{a}{2}\left(z+z^{-1}\right)+\frac{b}{2i}\left(z-z^{-1}\right)\right)^{2n}\frac{dz}{iz}\\
&=\frac{1}{2^{2n} i}\oint_{|z|=1}\frac{\left[(a-ib)z^2+(a+ib)\right]^{2n}}{z^{2n+1}}dz\\
&=2\pi\frac{1}{2^{2n}}\binom{2n}{n}(a-ib)^n(a+ib)^n,
\end{align*}
using Cauchy's integral formula.
\end{proof}

 \begin{rem} The observable $H=\Op^{\hbar}(f)$ with $f(x,p)=(ax+bp)^n$, when $n=1$ and $a=0$ corresponds to 
 $H=\hat{p}$, and  $H_N=\Pi_N \hat{p}\Pi_N$ is the \emph{truncated momentum operator}  on the subspace of the first $N$ eigenfunctions of the harmonic oscillator. This truncated observable corresponds to the Zeno Hamiltonian constructed in~\cite{Raimond10,Raimond12}, and studied in~\cite{Cunden23}. The convergence of the symbols $\sigma^{\hbar}_{\Pi_N}$ and $\sigma^{\hbar}_{\Pi_N\hat{p}\Pi_N}$ was proved in~\cite{Cunden23} in a different topology. 
More precisely, the following space of test functions introduced by Lions and Paul~\cite{Lions93},
\be
\AA=\left\{f\in C_0(\R_x\times\R_p)\colon \FF_2 f\in L^1\left(\R_y;C_0(\R_x)\right) \right\},
\ee
equipped with the norm
$\|f\|_{\AA}:=\int_{\R_y} \sup_x\left|\FF_2 f(x,y)\right| dy$ was considered. The normed space $\AA$ is a separable Banach algebra. Let $\AA'$ be the its dual. In~\cite{Cunden23} it was proved that the symbols $(\sigma^{\hbar}_{\Pi_N})_{N\geq1}\subset\AA'$, and $(\sigma^{\hbar}_{\Pi_N\hat{p}\Pi_N})_{N\geq1}\subset\AA'$ converge in the weak$^*$-topology, i.e., 
 \begin{equation}
 \label{eq:asymp_sigma_Pi_NAprime}
 \lim_{\substack{N\to\infty, \hbar\to0\\\hbar N=\mu}}\langle\sigma^{\hbar}_{\Pi_N}-\chi_{D},\varphi\rangle_{\AA', \AA}=0,\quad 
\lim_{\substack{N\to\infty, \hbar\to0\\\hbar N=\mu}}\langle\sigma^{\hbar}_{\Pi_N\hat{p}\Pi_N}-p\chi_{D},\varphi\rangle_{\AA', \AA}=0,
\end{equation}  
for all $\varphi\in\AA$, where $\langle \cdot, \cdot \rangle_{\AA', \AA}$ denotes the pairing between $\AA' $ and $\AA$.
When trying to extend these results to other models we faced the problem that, while $\AA$ and $\AA'$ are well-suited for some scopes, their are not so for the general problem of truncated quantum observables. In particular,  it is more natural to check $\|\sigma_{\Pi_N^{\perp}H\Pi_N}^{\hbar}\|_{L^2}\to 0$,
rather than $\|\sigma_{\Pi_N^{\perp}H\Pi_N}^{\hbar}\|_{\AA'}\to0$.

 \end{rem}

\subsection{Proofs of the main theorems}\label{sebsect:proof}
We will repeatedly use some classical facts about convergence in Hilbert spaces. In the following $\HH$ and $\KK$ denote separable Hilbert spaces with the corresponding scalar products $\langle \cdot, \cdot \rangle_{\mathcal{H}}, \langle \cdot, \cdot \rangle_{\mathcal{K}}$ and norms $\| \cdot \|_{\mathcal{H}}, \| \cdot \|_{\mathcal{K}}$, see for instance~\cite{ReedSimon}.
\begin{lem}
\label{lem:basic}
Let $(\sigma_N)_{N\geq1}\subset \HH$, $\sigma\in\HH$, and $T\colon\HH\to\KK$ an isometry.
\begin{enumerate}
\item  If $(\sigma_N)$ is a bounded sequence and $\langle \sigma_N,g\rangle_{\HH}\to\langle \sigma,g\rangle_{\HH}$ for all $g$ in a dense subspace $\BB$ of $\HH$, then $\sigma_N$ weakly converges to $\sigma$, i.e. $\langle \sigma_N,\varphi\rangle_{\HH}\to\langle \sigma,\varphi\rangle_{\HH}$  for all $\varphi\in \HH$;
\item If $\sigma_N$ weakly converges to $\sigma$, and $\|\sigma_N\|_{\HH}\to \|\sigma\|_{\HH}$, then $\sigma_N$ strongly converges to $\sigma$, i.e. $\|\sigma_N-\sigma\|_{\HH}\to0$;
\item $\sigma_N$ (weakly) converges to $\sigma$ in $\HH$ if and only if $T\sigma_N$ (weakly) converges to $T\sigma$ in $\KK$.
\end{enumerate}
\end{lem}

Let $\HH=L^2(\R_x\times\R_p)$ and $\KK=L^2(\R_x\times\R_y)$. The partial Fourier transform $\FF_2\colon \HH\to \KK$ is (up to a factor $2\pi$) an isometry. Consider $\sigmaPc$ and $\chi_{\Omega}$ in $\HH$, and  their images $\FF_2 \sigmaPc$ and $\FF_2 \chi_{\Omega}$ in $\KK$.
We report their explicit formulae.
\begin{lem} For all $N\geq1$, $\hbar>0$, and $\mu>0$,
\begin{align}
\FF_2 \sigmaPc(x,y)&=\left(2\pi\hbar\right) K_{\Pi_N}\left(x-\frac{\hbar y}{2},x+\frac{\hbar y}{2}\right),\\
\FF_2 \chi_{\Omega}\left(x,y\right)&=2\sqrt{2\left(E(\mu)-V(x)\right)_+}S(2\sqrt{2\left(E(\mu)-V(x)\right)_+}y),
\end{align}
where $K_{\Pi_N}$ is the integral kernel of $\Pi_N$, and $S$ is the \emph{sine kernel} defined in~\eqref{eq:sinekernel}.
\end{lem}
The key ingredient in the proof of Theorem~\ref{thm:main} is  the uniform convergence of $\FF_2 \sigmaPc$ to $\FF_2 \chi_{\Omega}$ on compact sets.

\begin{thm}\label{thm:convFcompact}
For any compact sets $U\Subset \R_x$,$V\Subset\R_y$, there exist constants $\hbar_0>0$ and $C>0$ (both depending on $U,V$), such that
\begin{align}
\label{eq:estimate_conv}
&\sup_{x\in U}\sup_{ y\in V}\left|\FF_2\sigmaPc(x,y)-\FF_2 \chi_{\Omega}\left(x,y\right)\right|\leq C\hbar,
\end{align}
for all $\hbar<\hbar_0$. 
\end{thm}
\begin{proof}
The estimate~\eqref{eq:estimate_conv}, is an adapted restatement of the uniform convergence of the rescaled kernel $K_{\Pi_N}$ to the sine kernel on compact sets, proved in~\cite[Sec.~3]{Deleporte21}. 
\end{proof}
\begin{cor} \label{corollaryconvF} For all $\mu>0$, 
\be
\lim_{\substack{N\to\infty, \hbar\to0\\\hbar N=\mu}}\|\FF_2\sigmaPc-\FF_2 \chi_{\Omega}\|_{L^2(\R_x\times\R_y)}=0.
\ee
\end{cor}
\begin{proof} We simply observe that for all $N$, $\|\FF_2\sigmaPc\|^2_{\KK}=\hbar N$ and $\|\FF_2\chi_{\Omega}\|^2_{\KK}=\mu$. So, with $\hbar=\mu/N$, we have that $(\FF_2\sigmaPc)_{N\geq1}$ is bounded in $\KK$ and $\|\FF_2\sigmaPc\|_{\KK}=\|\FF_2\chi_{\Omega}\|_{\KK}$. Therefore we only need to prove weak convergence.
Let $\BB=C_c(\R_x\times\R_y)\subset\KK$
be the subspace of continuous functions $g$ with compact support in $\R_x\times\R_y$. It is a standard fact that $\BB$ is dense in $\KK$.  The estimate~\eqref{eq:estimate_conv} with $\hbar=\mu/N$ implies that $\langle\FF_2 \sigmaPc,g\rangle_{\mathcal{K}}\to \langle\FF_2 \chi_{\Omega},g\rangle_{\mathcal{K}}$ for all $g\in \BB$. By density this proves weak convergence of $\FF_2 \sigmaPc$ to $\FF_2 \chi_{\Omega}$. 
\end{proof}
\begin{proof}[Proof of Theorem~\ref{thm:main}]  Since $\FF_2\colon \HH\to \KK$ is (up to a factor $2\pi$) an isometry, the convergence of $\sigmaPc$ to $\chi_{\Omega}$ in $\HH$ follows from the convergence of $\FF_2\sigmaPc$ to $\FF_2\chi_{\Omega}$ in $\KK$ obtained in Corollary~\ref{corollaryconvF}.  
\end{proof}
\begin{proof}[Proof of Theorem~\ref{thm:main2}] Let $f(x,p)=f_1(x)+f_2(x)$. 
By Condition~\eqref{item:cond2_gen}, 
it is enough to prove that $$\langle \sigma_{\Pi_N^{}\Op^{\hbar}(f)\Pi_N}^{\hbar}-f\chi_{\Omega},g\rangle_{\HH}\to 0,\quad\text{ for all $g\in\BB$},$$  where $\BB$ is some dense subspace of $\HH$.  We write, 
\begin{multline*}
\langle \sigma_{\Pi_N^{}\Op^{\hbar}(f)\Pi_N}^{\hbar}-f\chi_{\Omega},g\rangle_{\HH}=\\
\underbrace{\langle\sigma_{\Pi_N\Op^{\hbar}(f)\Pi_N}^{\hbar}-\sigma_{f(\hat{x},\hat{p})\Pi_N}^{\hbar},g\rangle_{\HH}}_{I_1(N)}
+\underbrace{\langle\sigma_{\Op^{\hbar}(f)\Pi_N}^{\hbar}-f\sigma_{\Pi_N}^{\hbar},g\rangle_{\HH}}_{I_2(N)}
+\underbrace{\langle f\sigma_{\Pi_N}^{\hbar}-f\chi_{\Omega},g\rangle_{\HH}}_{I_3(N)}.
\end{multline*}

The first term can be majorized as
\begin{align*}
\left|I_1(N)\right|&\leq \|\sigma_{\Op^{\hbar}(f)\Pi_N}^{\hbar}-\sigma_{\Pi_N\Op^{\hbar}(f)\Pi_N}^{\hbar}\|_{\HH}\|g\|_{\HH}\\
&=\|\sigma_{\Op^{\hbar}(f)\Pi_N-\Pi_N\Op^{\hbar}(f)\Pi_N}^{\hbar}\|_{\HH}\|g\|_{\HH}\\
&=\|\sigma_{\Pi_N^{\perp}\Op^{\hbar}(f)\Pi_N}^{\hbar}\|_{\HH}\|g\|_{\HH},
\end{align*}
and tends to zero by~\eqref{item:cond3_gen}.

We now proceed to analyse $I_2(N)$ and $I_3(N)$.
By linearity we can consider separately the cases $f_2 \equiv 0$ and $f_1 \equiv 0$. We start from the case $f_2 \equiv 0$, and we write $\Op^{\hbar}(f)=f_1(\hat{x})$. Choose $\BB_1=C_c(\R_x;\FF^{-1}(C_c(\R_y)))\subset\HH$. Let $g\in\BB_1$. The support of $\FF_2 g\in\KK$ is contained in $U\times V$ for some $U\Subset \R_x$, and $V\Subset \R_y$.

For $I_2(N)$ we first note that
\begin{align*}
\sigma_{\Op^{\hbar}(f)\Pi_N}^\hbar (x,p) &= \sigma_{f_1(\hat{x})\Pi_N}^{\hbar}(x,p)=\int_{\R_y}\hbar f_1\left(x-\frac{\hbar y}{2}\right)K_{\Pi_N}\left(x-\frac{\hbar y}{2},x+\frac{\hbar y}{2}\right)e^{ipy}dy,\\
f(x,p)\sigma_{\Pi_N}^{\hbar}(x,p)&= f_1(x)\sigma_{\Pi_N}^{\hbar}(x,p)\int_{\R_y}\hbar f_1\left(x\right)K_{\Pi_N}\left(x-\frac{\hbar y}{2},x+\frac{\hbar y}{2}\right)e^{ipy}dy.
\end{align*}
Using Cauchy-Schwarz, we get
\begin{align*}
|I_2(N)|
&\leq\int_{\R_x\times\R_y} \left|\left(f_1\left(x-\frac{\hbar y}{2}\right)-f_1\left(x\right)\right)\hbar K_{\Pi_N}\left(x-\frac{\hbar y}{2},x+\frac{\hbar y}{2}\right)\FF_2g(x,y)\right|dxdy
\\
&\leq\sqrt{\mu}\left(\int_{U\times V} \left|\left(f_1\left(x-\frac{\hbar y}{2}\right)-f_1\left(x\right)\right)\FF_2g(x,y)\right|^2dxdy\right)^{\frac{1}{2}}\to0,
\end{align*}
by dominated convergence.

The last term is easily bounded
$$
|I_3(N)|=\left|\langle f\sigma_{\Pi_N}^{\hbar}-f\sigma,g\rangle_{\HH}\right|=\left|\langle \sigma_{\Pi_N}^{\hbar}-\sigma,fg\rangle_{\HH}\right|\leq \sup_{x\in U}|f_1(x)|\left|\langle \sigma_{\Pi_N}^{\hbar}-\sigma,g\rangle_{\HH}\right|
$$
since $g(\cdot,p)$ has support contained in $U\Subset\R_x$. By Condition~\eqref{item:cond1_gen}, we have  $|I_3(N)|\to0$.
\par

We now consider the case  $f_1 \equiv 0$, hence $ \Op^{\hbar}(f) = f_2(\hat{p}) $ . Choose $\BB_2=\FF C_c(\R_q;C_c(\R_p))\subset\HH$. Let $g\in\BB_2$. The support of $\FF_1^{-1} g$ is contained in $U\times V$ for some $U\Subset \R_q$, and $V\Subset \R_p$.

We notice that
\begin{align*}
\sigma_{\Op^{\hbar}(f)\Pi_N}^\hbar (x,p) &=\sigma_{f_2(\hat{p})\Pi_N}^{\hbar}(x,p)=\frac{1}{2\pi}\int_{\R_y} f_2\left(p-\frac{\hbar q}{2}\right)\widehat{K}_{\Pi_N}\left(\frac{1}{\hbar}p-\frac{q}{2},\frac{1}{\hbar}p+\frac{q}{2}\right)e^{iqx}dx,\\
f(x,p)\sigma_{\Pi_N}^{\hbar}(x,p) &=f_2(p)\sigma_{\Pi_N}^{\hbar}(x,p)=\frac{1}{2\pi}\int_{\R_y} f_2\left(p\right)\widehat{K}_{\Pi_N}\left(\frac{1}{\hbar}p-\frac{q}{2},\frac{1}{\hbar}p+\frac{q}{2}\right)e^{iqx}dx.
\end{align*}
So, proceeding as above we get 
\begin{align*}
|I_2(N)|
&\leq2\pi \int_{\R_q\times\R_p} \left|\left(f_2\left(p-\frac{\hbar q}{2}\right)-f_2\left(p\right)\right) \widehat{K}_{\Pi_N}\left(\frac{1}{\hbar}p-\frac{q}{2},\frac{1}{\hbar}p+\frac{q}{2}\right)\FF_1^{-1}g(q,p)\right|dqdp
\\
&\leq\sqrt{\mu}\left(\int_{U\times V} \left|\left(f_2\left(p-\frac{\hbar q}{2}\right)-f_2\left(p\right)\right)\FF_1^{-1}g(q,p)\right|^2dqdp\right)^{\frac{1}{2}}\to0,
\end{align*}
by dominated convergence.

For $I_3(N)$ we proceed as above
$$
|I_3(N)|=\left|\langle f\sigma_{\Pi_N}^{\hbar}-f\chi_{\Omega},g\rangle_{\HH}\right|=\left|\langle \sigma_{\Pi_N}^{\hbar}-\chi_{\Omega},fg\rangle_{\HH}\right|\leq \sup_{p\in V}|f_2(p)|\left|\langle \sigma_{\Pi_N}^{\hbar}-\chi_{\Omega},g\rangle_{\HH}\right|
$$
since $g(x,\cdot)$ has support contained in $V\Subset\R_p$. Again by~\eqref{item:cond1_gen}, we have  $|I_3(N)|\to0$.

\end{proof}

\begin{proof}[Proof of Theorem~\ref{thm:main3}] The scheme of the proof is the same as for the proof of Theorem~\ref{thm:main2}. Let $\BB=\SS(\R_x\times\R_p)\subset\HH$ the space of Schwartz functions. Then, by Condition~\eqref{item:cond2_gen}, it is enough to show that for all $g\in \BB$, 
\begin{multline*}
\langle \sigma_{\Pi_N^{}\Op^{\hbar}(f)\Pi_N}^{\hbar}-f\chi_{\Omega},g\rangle_{\HH}=\\
\underbrace{\langle\sigma_{\Pi_N\Op^{\hbar}(f)\Pi_N}^{\hbar}-\sigma_{\Op^{\hbar}(f)\Pi_N}^{\hbar},g\rangle_{\HH}}_{I_1(N)}
+\underbrace{\langle\sigma_{\Op^{\hbar}(f)\Pi_N}^{\hbar}-f\sigma_{\Pi_N}^{\hbar},g\rangle_{\HH}}_{I_2(N)}
+\underbrace{\langle f\sigma_{\Pi_N}^{\hbar}-f\chi_{\Omega},g\rangle_{\HH}}_{I_3(N)}.
\end{multline*}
tends to $0$. Again, by Condition~\eqref{item:cond3_gen}, $I_1(N)\to0$. 

The last term is
\be
I_3(N)=\langle f\sigma_{\Pi_N}^{\hbar}-f\sigma,g\rangle_{\HH}=\langle \sigma_{\Pi_N}^{\hbar}-\sigma,fg\rangle_{\HH}.
\ee
Since $f$ has at most polynomial growth, $fg\in \HH$ and therefore $I_3(N)\to0$ 
by~\eqref{item:cond1_gen}.

It remains to bound $I_2(N)$. We have $\sigma^{\hbar}_{\Op^{\hbar}(f)\Pi_N}(x,p)=\sigma^{\hbar}_{\Op^{\hbar}(f)} \,\sharp \, \sigma^{\hbar}_{\Pi_N}(x,p)=f \,\sharp\,\sigma^{\hbar}_{\Pi_N}(x,p)$. Using the explicit formula for the Moyal product~\eqref{eq:Moyal2},
\begin{align}
\label{eq:repr1prime}
\sigma^{\hbar}_{\Op^{\hbar}(f)\Pi_N}(x,p)
&=\int_{\R^4} e^{-i(k_1x_2-k_2x_1)}f\left(x-x_1,p-\frac{\hbar k_1}{2}\right)\sigma^{\hbar}_{\Pi_N}\left(x-x_2,p-\frac{\hbar k_2}{2}\right)\frac{dx_1 dk_1}{2\pi }\frac{dx_2 dk_2}{2\pi }.
\end{align}

It is useful to first consider the case of $f(x,p)$ a real-valued polynomial in both $x$ and $p$.
In such a case, by Taylor's formula,
\be
f\left(x-x_1,p-\frac{\hbar k_1}{2}\right)=\sum_{0\leq \ell_x+\ell_p\leq d}\frac{1}{\ell_x!\ell_p!}\left[\partial^{\ell_x}_x\partial^{\ell_p}_p f\left(x,p\right)\right]\left(-x_1\right)^{\ell_x}\left(-\frac{\hbar k_1}{2}\right)^{\ell_p},
\ee
where $d\geq0$ is the (total) degree of $f$.
Therefore,
\be
\label{eq:Moyal_polypoly}
\langle\sigma^{\hbar}_{\Op^{\hbar}(f)\Pi_N},g\rangle_{\HH}=\sum_{0\leq \ell_x+\ell_p\leq d}\frac{(-1)^{\ell_x+\ell_p}}{\ell_x!\ell_p!}J_{\ell_x,\ell_p}, 
\ee
where, for all $0\leq \ell_x+\ell_p\leq d$, 
\begin{multline*}
J_{\ell_x,\ell_p}=\left(\frac{\hbar}{2}\right)^{\ell_p}\int_{\R^6}x_1^{\ell_x}k_1^{\ell_p}e^{-i(k_1x_2-k_2x_1)}\left[\partial^{\ell_x}_x\partial^{\ell_p}_p f\left(x,p\right)\right]\sigma^{\hbar}_{\Pi_N}\left(x-x_2,p-\frac{\hbar k_2}{2}\right)\\\times g(x,p)\frac{dx_1 dk_1}{2\pi }\frac{dx_2 dk_2}{2\pi }dxdp\\
=\left(\frac{\hbar}{2}\right)^{\ell_p}(-i)^{\ell_x}i^{\ell_p}\int_{\R^6}\partial_{k_2}^{\ell_x}\partial_{x_2}^{\ell_p}\left[e^{-i(k_1x_2-k_2x_1)}\right]\left[\partial^{\ell_x}_x\partial^{\ell_p}_p f\left(x,p\right)\right]\sigma^{\hbar}_{\Pi_N}\left(x-x_2,p-\frac{\hbar k_2}{2}\right)\\\times g(x,p)\frac{dx_1 dk_1}{2\pi }\frac{dx_2 dk_2}{2\pi }dxdp\\
=\left(\frac{\hbar}{2}\right)^{\ell_p}(-i)^{\ell_x}i^{\ell_p}\int_{\R^6}e^{-i(k_1x_2-k_2x_1)}\left[\partial^{\ell_x}_x\partial^{\ell_p}_p f\left(x,p\right)\right]\partial_{k_2}^{\ell_x}\partial_{x_2}^{\ell_p}\left[\sigma^{\hbar}_{\Pi_N}\left(x-x_2,p-\frac{\hbar k_2}{2}\right)\right]\\\times g(x,p)\frac{dx_1 dk_1}{2\pi }\frac{dx_2 dk_2}{2\pi }dxdp\\
=\left(\frac{\hbar}{2}\right)^{\ell_x+\ell_p}(-i)^{\ell_x}i^{\ell_p}\int_{\R_x\times\R_p}\sigma^{\hbar}_{\Pi_N}\left(x,p\right)\left[\partial_{p}^{\ell_x}\partial_{x}^{\ell_p}\left[\partial^{\ell_x}_x\partial^{\ell_p}_p f\left(x,p\right)\right]g(x,p)\right]dxdp. 
\end{multline*}
Hence,
\begin{align*} 
\left|J_{\ell_x,\ell_p}\right|&\leq \left(\frac{\hbar}{2}\right)^{\ell_x+\ell_p}\|\sigma^{\hbar}_{\Pi_N}\|_{\HH}\left\|\partial_{p}^{\ell_x}\partial_{x}^{\ell_p}\left(\partial^{\ell_x}_x\partial^{\ell_p}_p f\right)g\right\|_{\HH}\\&=\left(\frac{\hbar}{2}\right)^{\ell_x+\ell_p}\sqrt{2\pi\mu}\left\|\partial_{p}^{\ell_x}\partial_{x}^{\ell_p}\left(\partial^{\ell_x}_x\partial^{\ell_p}_p f\right)g\right\|_{\HH}.
\end{align*} 
We conclude that in~\eqref{eq:Moyal_polypoly}, the only summand that survives in the limit is  the one with $\ell_x=\ell_p=0$:
\begin{align}
 \lim_{\substack{N\to\infty, \hbar\to0\\\hbar N=\mu}}\langle\sigma^{\hbar}_{\Op^{\hbar}(f)\Pi_N},g\rangle_{\HH}
 =J_{0,0}=\langle f\sigma^{\hbar}_{\Pi_N},g\rangle_{\HH},
\end{align}
and so $I_2(N)\to0$. 
\par

We now consider the general case of $f$ polynomial in $p$,
\be
f\left(x,p\right)=\sum_{\ell=0}^d\frac{1}{\ell!}a_{\ell}(x)p^{\ell},
\ee
with polynomially bounded coefficients:
\be
\label{eq:bound_coeff_a}
\left|a_{\ell}(x)\right|\leq C_{\ell}x^{2\alpha_{\ell}},
\ee
for some constants $C_{\ell}\geq0$ and positive integers $\alpha_{\ell}\geq1$.
Then,
\be
\label{eq:Moyal_poly_p}
\langle\sigma^{\hbar}_{\Op^{\hbar}(f)\Pi_N},g\rangle_{\HH}=\sum_{0\leq \ell\leq d}\frac{1}{\ell!}J_{\ell}, 
\ee
where
\begin{align*}
J_{\ell}=&\left(-\frac{i\hbar}{2}\right)^{\ell}\int\limits_{\R_x\times\R_p\times\R_{x_1}}\int\limits_{\R_{k_2}}\left(e^{ik_2x_1}\sigma^{\hbar}_{\Pi_N}\left(x,p-\frac{\hbar k_2}{2}\right)\frac{dk_2}{2\pi }\right)\partial_x^{\ell}\partial_p^{\ell} f(x-x_1,p)g(x,p) \,dxdpdx_1
\end{align*}
From~\eqref{eq:Moyal_poly_p}, we have
\be
\partial_x^j\partial_p^k f(x,p)=\sum_{\ell=0}^{d-k} \frac{1}{\ell!}\left[\partial_x^ja_{\ell}(x)\right]p^{\ell},
\ee
and therefore, by Leibniz's rule, we have now
\be
J_{\ell}=\left(-\frac{i\hbar}{2}\right)^{\ell}\sum_{j=0}^{\ell}\sum_{s=0}^{d-\ell}\binom{\ell}{j}\frac{1}{s!}I_{\ell,s,j}
\ee
where
\begin{align*}
I_{\ell,s,j}&=\int\limits_{\R_x\times\R_p\times\R_{x_1}}\int\limits_{\R_{k_2}}\left(e^{ik_2x_1}\sigma^{\hbar}_{\Pi_N}\left(x,p-\frac{\hbar k_2}{2}\right)\frac{dk_2}{2\pi }\right)\left[\partial_x^ja_s(x-x_1)\right]p^{s}\partial_x^{\ell-j}g(x,p)dxdpdx_1\\
&=\int\limits_{\R_x\times\R_p\times\R_{y}}\int\limits_{\R_{k}}\left(e^{-iy(k-p)}\sigma^{\hbar}_{\Pi_N}\left(x,k\right)\frac{d k}{2\pi }\right)\left[\partial_x^ja_s\left(x-\frac{\hbar y}{2}\right)\right]p^{s}\partial_x^{\ell-j}g(x,p)dxdpdy\\
&=\int\limits_{\R_x\times\R_{y}}\hbar K_{\Pi_N}\left(x-\frac{\hbar y}{2},x+\frac{\hbar y}{2}\right)\left[\partial_x^ja_s\left(x-\frac{\hbar y}{2}\right)\right]\left(\int_{\R_p}p^{s}\partial_x^{\ell-j}g(x,p)e^{-iyp}dp\right)dxdy\\
\end{align*}
Let $F_{\ell,s,j}(x,p):=p^{s}\partial_x^{\ell-j}g(x,p)$. 
Note that if $g\in\SS(\R_x\times\R_p)$, then $\FF_2(F_{\ell,s,j})\in\SS(\R_x\times\R_y)$ for all $\ell,s,j$.
From~\eqref{eq:bound_coeff_a}, we have
$$\left|\partial_x^ja_s\left(x-\frac{\hbar y}{2}\right)\right|\leq C_{s,j}\left(x-\frac{\hbar y}{2}\right)^{2\alpha_{s}},$$
for some constant $C_{s,j}\geq0$.
Therefore 
\be
\left|I_{\ell,s,j}\right| \leq2\pi\hbar \|K_{\Pi_N}\|_{\KK}\left(\int_{\R_x\times\R_y}\left|\left[\partial_x^ja_s\left(x-\frac{\hbar y}{2}\right)\right]\FF_2(F_{\ell,s,j})(x,y)\right|^2dxdy\right)^{1/2}<\infty.
\ee
We deduce that in~\eqref{eq:Moyal_poly_p}, the only summand that survives in the limit is the one with $\ell=0$, and this concludes the proof.
\end{proof}

\section{Truncated observables for a free particle in a box} 
\label{sec:box}
We now study the family of spectral projections of the Schr\"odinger operator  $A=-\frac{\hbar^2}{2}\frac{d^2}{dx^2}$ acting as
\be
Af(x)=-\frac{\hbar^2}{2}f''(x),
\ee
for all $f\in D(A)=H^2_0([-L,L])$. The self-adjoint operator $A$ is the Hamiltonian of a free particle confined in the one-dimensional box $[-L,L]$ with  zero Dirichlet boundary conditions. The eigenvalues of $A$ are
\begin{equation}
\label{eq:eigenvalues}
E_k=\frac{\hbar^2}{2}\left( \frac{k\pi}{2L}\right)^2.
\end{equation}
The corresponding normalised eigenfunctions
\begin{equation}
u_k(x)=\frac{\chi_{[-L,L]}(x)}{\sqrt{L}}\sin\left( \frac{k\pi}{2L}(x+L)\right), \quad k \geq 1,
\end{equation}
 form an orthonormal basis in $L^2([-L,L])$. A useful formula is
\be
\label{eq:Cheby}
U_k(\cos x)=\frac{\sin(k+1)x}{\sin x},
\ee
where $U_k$'s are the Chebyshev polynomials of the second kind ($U_0(y)=1$, $U_1(y)=2y$ and $U_{k}(y)=2y U_{k-1}(y)-U_{k-2}(y)$, for all $k \geq 2$). Using~\eqref{eq:Cheby}, we can rewrite the eigenfunctions of  $A$ as
\be
u_k(x)=\frac{\chi_{[-L,L]}(x)}{\sqrt{L}}\sin\left(\frac{\pi}{2L}(x+L)\right)U_{k-1}\left(\cos\left(\frac{\pi}{2L}(x+L)\right)\right).
\ee
Throughout this section we consider the nested sequence of orthogonal projections
\begin{equation}
\label{eq:proj_box}
\Pi_{N}:=\chi_{(-\infty,E_N]}(A)=\sum_{k=1}^N |u_k\rangle \langle u_k|,\quad N\geq1.
\end{equation}

\par

Note that for the Sch\"odinger operator of a particle in a box, the confining potential $V(x)$ is formally infinite outside the box and thus this case is not covered by Theorem~\ref{thm:main}.

In this section we will repeatedly use the following identity (an application of trigonometric formulae and integration by parts).
\begin{lem}
\label{lem:symbol_ujuk} For all $j,k\geq1$, 
\begin{multline}
\label{eq:symbol_ujuk}
\sigma_{|u_j\rangle\langle u_k|}^{\hbar}(x,p):=\int_{\R_{y}}\hbar u_j\left(x-\frac{\hbar y}{2}\right)u_k\left(x+\frac{\hbar y}{2}\right)e^{i p y} dy\\
 =\chi_{[-L,L]}(x)\frac{\hbar}{2 L}
 \sum_{\epsilon_1,\epsilon_2=\pm 1}
\frac{\epsilon_1 e^{-i\epsilon_2\frac{\pi}{2L}(j-\epsilon_1k)(x+L)} }{\frac{\pi\hbar}{4L}(j+\epsilon_1k)+\epsilon_2p}\sin\left( \frac{2}{\hbar}(L-|x|)\left( \frac{\pi\hbar}{4L}(j+\epsilon_1k)+\epsilon_2p\right)\right).
\end{multline}
\end{lem}

The starting point of the analysis are the explicit formulae for the integral kernel and the symbol of $\Pi_N$.
\begin{prop}[Integral kernel and symbol of $\Pi_N$]
\begin{align}
K_{\Pi_N}(x,y)&=\sum_{k=1}^{N}u_k(x)u_k(y)\nonumber\\
&=\frac{{\chi_{[-L,L]^2}(x,y)}}{4L}\left[D_{N}\left(\frac{\pi}{2L}(x-y)\right)-D_{N}\left(\frac{\pi}{2L}(x+y+2L\right)\right].
\label{eq:K_N1}
\end{align}
\begin{align}
\label{eq:symbol_P_N}
    \sigma_{\Pi_N}^\hbar(x,p)=\chi_{[-L,L]}(x)\frac{\hbar}{2 L}
   \sum_{k=1}^N&\Biggl[\frac{\sin \left(\frac{2}{\hbar} (L-\left| x\right| ) \left(\frac{\hbar\pi  k}{2 L}+p\right)\right)}{\frac{\hbar\pi  k}{2 L}+p}+\frac{\sin \left(\frac{2}{\hbar} (L-\left| x\right| ) \left(\frac{\hbar\pi  k}{2 L}-p\right)\right)}{\frac{\hbar\pi  k}{2 L}-p}\nonumber \\
    &-2\frac{\cos \left(\frac{\pi  k}{L}\left(L+x\right)\right) \sin \left(\frac{2}{\hbar} (L-\left| x\right| )p\right)}{p}\Biggr].
    \end{align}
\end{prop}
\begin{proof}
Formula \eqref{eq:K_N1} can be obtained  by writing the complex exponential form of the $\sin(\cdot)$ function, and then using the geometric sum. 
Formula \eqref{eq:symbol_P_N} follows from Lemma~\ref{lem:symbol_ujuk} upon observing that $\sigma_{\Pi_N}^\hbar=  \sum_{k=1}^N \sigma_{|u_k\rangle\langle u_k|}^\hbar$.
\end{proof}

\subsection{Semiclassical macroscopic limit of $ \sigma_{\Pi_N}^\hbar$}

In the semiclassical limit $N\to\infty$, $\hbar\to0$ with the product $\hbar N=\mu$ kept fixed, we again expect the symbol to concentrate in the classically allowed region. A natural candidate as classical counterpart of the quantum Hamiltonian $A=-\frac{\hbar^2}{2}\frac{d^2}{dx^2}$, $D(A)=H^2_0([-L,L])$, is the function
$$
a(x,p)=\frac{1}{2}p^2,\quad (x,p)\in [-L,L]\times \R_p
$$ 
In the limit, the classically allowed region corresponding to the  subspace $\operatorname{Ran}\Pi_N$ with $\Pi_{N}=\chi_{(-\infty,E_N]}(A)$, and $E_N=\frac{\hbar^2N^2}{2}\left( \frac{\pi}{2L}\right)^2$ is the rectangle $\Omega=R$ given by
\be
\begin{aligned}
\label{eq:rectangle}
R&=\left\{(x,p)\in\R_x\times \R_p\colon x\in [-L,L]\quad\text{and}\quad a(x,p) \leq  \frac{1}{2}\frac{\mu^2\pi^2}{4L^2}\right\}\\
&=\left\{(x,p)\in\R_x\times \R_p\colon -L\leq x\leq L\quad\text{and}\quad -\frac{\mu\pi}{2L}\leq p\leq  \frac{\mu\pi}{2L}\right\},
\end{aligned}
\ee
whose boundary is 
$$\partial R:=\left\{(x,p) \in \R_x\times \R_p : |x|=L\right\}\cap \left\{(x,p) \in \R_x\times \R_p : |p|=\frac{\mu \pi}{{2L}}\right\}.$$

The first  result of this Section is the identification of the limit of the Weyl symbols $\sigmaPc(x,p)$  of the projection operator $\Pi_{N}$.  Let 
\be
\label{eq:chiOmega}
\chi_{R}(x,p)=
\chi_{[-L,L]}(x)\chi_{\left[-\frac{\mu \pi}{{2L}},\frac{\mu \pi}{{2L}}\right]}(p)
\ee 
be the characteristic function of the rectangle~\eqref{eq:rectangle}.

 \begin{thm}[$L^2$-convergence of the symbol of $\Pi_N$]\label{thm:convergenceL2_R} Set $\mu>0$. Then,
 \begin{align}
\label{eq:asymp_sigma_P_<N}
\lim_{\substack{N\to\infty, \hbar\to0\\\hbar N=\mu}}\|\sigmaPc-\chi_{R}\|_{L^2(\R_x\times\R_p)}=0, 
\end{align}
with $\chi_{R}$ as in~\eqref{eq:chiOmega}.
\end{thm}
For a numerical illustration, see Figure~\ref{fig:symbol}. The proof scheme is the same as for Theorem~\ref{thm:main}. The only ingredient missing (that does not immediately follows from the general results of~\cite{Deleporte21}) is the uniform convergence of $\FF_2\sigmaPc$ to $\FF_2\chi_{R}$ on compact subsets of $\R_x\times\R_y$.  For completeness we report here the precise result and its proof.
\begin{prop}
\label{prop:limitKbulk}
Fix $\mu>0$ and let $\hbar=\mu/N$. Let $U\Subset (-L,L)$ and $V\Subset\R$ compact sets, and set $C_U=\max_{x\in U}|x|$, and  $C_V=\max_{y\in V}|y|$.
 Let $\hbar_0$ such that $0<\hbar_0<(L-C_U)/C_V$. Then, for all $\hbar <\hbar_0$ we have that:
 \be
 \label{eq:limitKbulk}
 \sup_{x\in U}\sup_{y\in V}\left|2\pi\hbar K_{\Pi_N}\left(x- \frac{\hbar y}{2},x+   \frac{\hbar y}{2}\right)-\chi_{[-L,L]}(x)\frac{\pi\mu}{L}S\left(\frac{\pi\mu}{L} y\right)\right|\leq C \hbar ,
 \ee
where  $S(\cdot)$ is the sine kernel of Eq.~\eqref{eq:sinekernel}, and
\be
\label{eq:constant_C_bound}
C=\frac{\pi}{2L}\left[\frac{L}{L-C_U}+\frac{\pi\mu}{2L} C_V+1\right].
\ee
\end{prop}

\begin{rem}
If $V=\{0\}$, we have $y=0$, and the claim is true for all $\hbar=\mu/N$:
$$
\left|2\pi\hbar K_{\Pi_N}\left(x,x\right)-\frac{\pi\mu}{L}S\left(0\right)\right|=\frac{\pi}{2L}\hbar.
$$
(So we formally set $\hbar_0=+\infty$.)
\end{rem}
\begin{proof}[Proof of Proposition~\ref{prop:limitKbulk}] We proceed to show~\eqref{eq:limitKbulk}. For all $x\in(-L,L)$, we denote by $I_{x,\hbar}$ the  interval
$$
I_{x,\hbar}:=\left\{y\in\R\colon\text{ $\left|x-\hbar \frac{y}{2}\right|\leq L$ and $\left|x+\hbar \frac{y}{2}\right|\leq L$}\right\}. 
$$
With this notation we have that the rescaled kernel $2\pi\hbar K_{\Pi_N}\left(x-\hbar y/2,x+ \hbar y/2\right)$ is zero whenever $y\notin I_{x,\hbar}$.  If $V=\{0\}$, the claim is immediate for all $\hbar=\mu/N$.
Suppose $V\neq\{0\}$.
 We have
$$
2\pi\hbar K_{\Pi_N}\left(x+\hbar \frac{y}{2},x+ \hbar  \frac{y}{2}\right)=
\frac{\pi\hbar}{2L}\chi_{I_{x,\hbar}}(y)\left[D_{N}\left(\frac{\hbar \pi y}{2L}\right)-D_{N}\left(\frac{\pi}{L}(x+L)\right)\right].
$$
Note that, for $x\in U$, 
$$
\left|x\pm\hbar\frac{y}{2}\right|\leq \left|x\right|+\hbar\left|\frac{y}{2}\right|\leq  \max_{x\in U}\left|x\right|+\frac{\hbar}{2}\max_{y\in V}\left|y\right|=  C_U+\frac{\hbar}{2}C_V\leq L,
$$
since $\hbar<\hbar_0$. So, we have $\chi_{I_{x,\hbar}}(y)=1$ for all $y\in V$. Moreover, by the elementary inequality $\sin\left(\frac{\pi}{2}\frac{L+x}{L}\right)\geq \frac{L-|x|}{L}$, for $|x|\leq L$, we get
\begin{multline*}
\left|
D_{N}\left(\frac{\pi}{L}(x+L)\right)
\right|=
\left|
\frac{\sin\left((2N+1)\frac{\pi}{2}\frac{L+x}{L}\right)}{\sin\left(\frac{\pi}{2}\frac{L+x}{L}\right)}\right|
\leq  
\frac{\left|\sin\left((2N+1)\frac{\pi}{2}\frac{L+x}{L}\right)\right|}{\frac{L-|x|}{L}}\leq \frac{L}{L-C_U}
\end{multline*}
Therefore,
\begin{multline*}
\left|2\pi\hbar K_{\Pi_N}\left(x+\hbar \frac{y}{2},x+ \hbar  \frac{y}{2}\right)-\frac{\pi\mu}{L}S\left(\frac{\pi\mu}{L} y\right)\right|
\leq \frac{ \pi \hbar }{2L}\frac{L}{L-C_U}
+\left| \frac{ \pi \hbar }{2L}D_{N}\left(\frac{\hbar \pi y}{2L}\right)-\frac{\pi\mu}{L}S\left(\frac{\pi\mu}{L} y\right)\right|.
\end{multline*}
Recall the following corollary of the Lagrange mean value theorem: if $f$ is differentiable in $(a,b)$, and we set $x_j=a+j\frac{(b-a)}{n}$, $j=0, \dots, n-1$, then,
\be
\left|\frac{(b-a)}{n}\sum_{j=0}^{n-1}f(x_j)-\int_{a}^bf(x)dx\right|\leq\frac{M_1(b-a)^2}{2n},\quad M_1=\sup_{x\in(a,b)}\left|f'(x)\right|.
\ee
Therefore we have
\begin{multline*}
\left| \frac{ \pi \hbar }{2L}D_{N}\left(\frac{\hbar \pi y}{2L}\right)-\frac{\pi\mu}{L}S\left(\frac{\pi\mu}{L} y\right)\right|=\left| \frac{ \pi \hbar }{2L}\sum_{k=-N}^{N}e^{i\frac{\hbar \pi  k y}{2L}}- \frac{ \pi  }{2L}\int_{-\mu}^{\mu}e^{i\frac{ \pi  k y}{2L}}dk\right|\\
\leq \frac{ \pi  }{2L}
\left(\left| \hbar \sum_{k=-N}^{N-1}e^{i\frac{\hbar \pi  k y}{2L}}- \int_{-\mu}^{\mu}e^{i\frac{ \pi  k y}{2L}}dk\right|+\hbar\right)\leq \frac{ \pi  }{2L}\left(\frac{\pi}{2L}\frac{|y|\mu^2}{N}+\hbar\right)\leq \frac{ \pi  }{2L}\left(\frac{\pi}{2L}\mu C_V+1\right)\hbar.
\end{multline*}

Putting everything together, we get the claim.

\end{proof}

\begin{figure}[t]
	\centering
	\includegraphics[width=.90\textwidth]{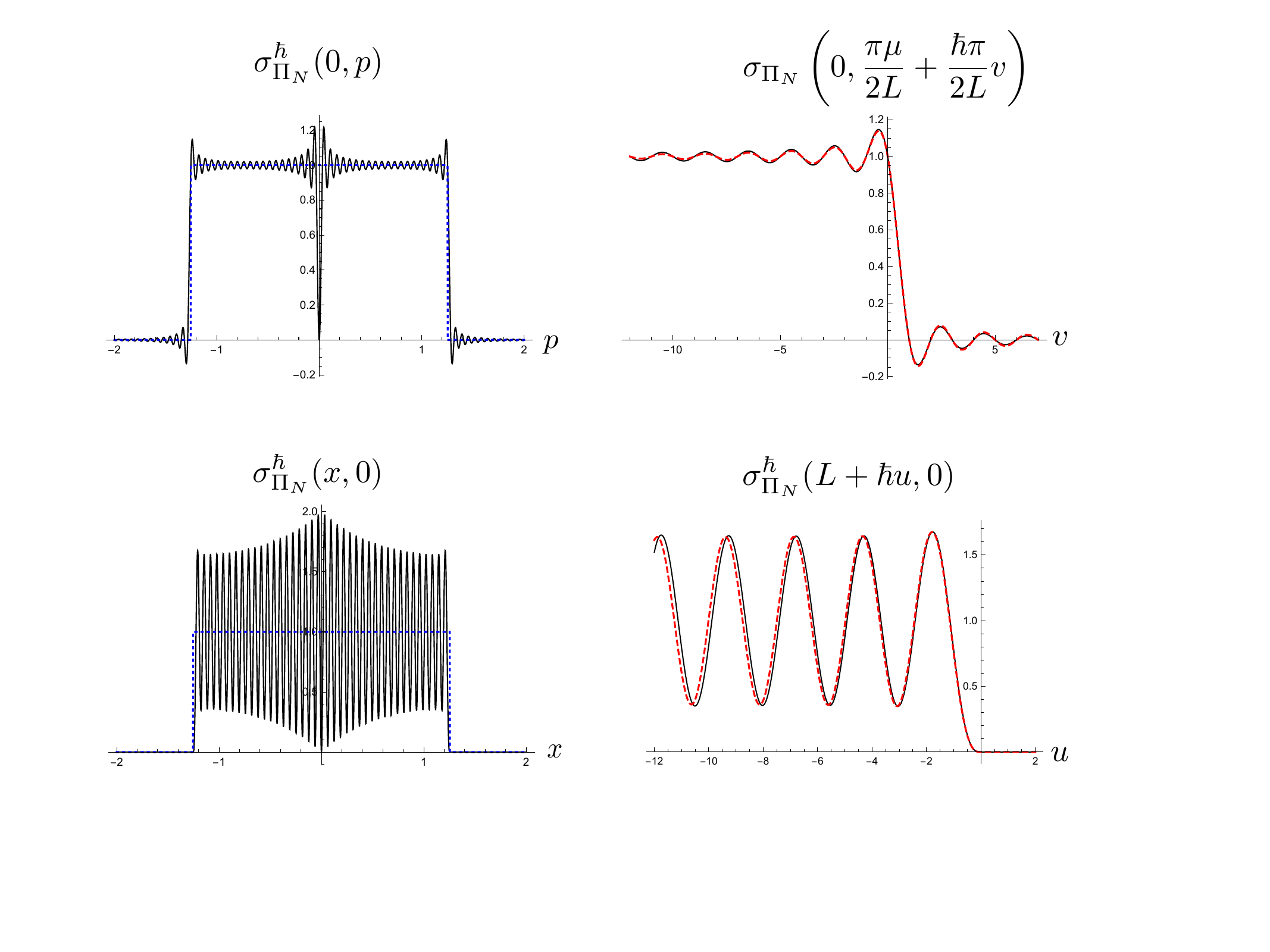} 
	\caption{Comparison between $\sigmaPc$ (solid black line), its $L^2$-limit $\chi_R$ in Eq.~\eqref{eq:chiOmega} (dotted blue line), and the microscopic pointwise limits in Eqs.~\eqref{eq:pointwise1}-\eqref{eq:pointwise2} (dashed red line).  Here $N=40$, $\mu=1$, and $L=\sqrt{\pi/2}$. }
	\label{fig:symbol}
\end{figure}

\subsection{Semiclassical pointwise limits of $\sigmaPc$}
\label{sub:pointwise}
At the boundary $\partial R$  of the classically allowed region, the symbol $\sigmaPc$  develops a jump, for large $N$. The second main result of this section concerns the \emph{pointwise} asymptotics of $\sigmaPc$near $\partial R$. By symmetry, it is enough to consider the left sides ($x=L$) and the upper side ($p=\frac{\pi\mu}{2L}$) of $\partial R$.

\begin{thm}[Pointwise convergence at microscopic scale near the boundary] 
\label{thm:pointwise}
The following  limits hold pointwise:
\begin{multline}
\label{eq:pointwise1}
\lim_{\substack{N\to\infty, \hbar\to0\\\hbar N=\mu}}\sigmaPc\left(L-\hbar u,p\right)=\\
\displaystyle \frac{\chi_{\R_+}(u)}{\pi} \left[\operatorname{Si}\left(2 u \left(p+\frac{\pi  \mu }{2 L}\right)\right)-\operatorname{Si}\left(2 u \left(p-\frac{\pi  \mu }{2 L}\right)\right)-\frac{\sin (2 p u)}{p u} \sin \left(\frac{\pi  \mu }{L} u\right)\right],
\end{multline}
where $\operatorname{Si}(\cdot)$ is the Sine integral function,
$$
\operatorname{Si}(x):=\int_0^x\frac{\sin t}{t}dt,
$$
and
\be
\label{eq:pointwise2}
\lim_{\substack{N\to\infty, \hbar\to0\\\hbar N=\mu}}\sigmaPc\left(x,\frac{\pi \mu}{2L}+\frac{\hbar\pi}{2L} {v}\right)= \displaystyle\frac{\chi_{[-L,L]}(x)}{2 L}\sum_{j=0}^{\infty}\frac{\sin \left( 2 (L-\left| x\right| )\frac{\pi }{2 L} (j+{v})\right)}{\frac{\pi }{2 L} (j+{v})},
\ee

\end{thm}
\begin{proof}
We begin by writing the symbol of $\Pi_N$ as
$$
\sigmaPc(x,p)=\chi_{[-L,L]}(x)\left(A+B-2C\right),
$$
where
\begin{align*}
A=   \frac{\hbar}{2 L}\sum_{k=1}^N&\frac{\sin \left(\frac{2}{\hbar} (L-\left| x\right| ) \left(\frac{\hbar\pi  k}{2 L}+p\right)\right)}{\frac{\hbar\pi  k}{2 L}+p},\quad
   B= \frac{\hbar}{2 L}  \sum_{k=1}^N\frac{\sin \left(\frac{2}{\hbar} (L-\left| x\right| ) \left(\frac{\hbar\pi  k}{2 L}-p\right)\right)}{\frac{\hbar\pi  k}{2 L}-p},\\
   C&= \frac{\hbar}{2 L}  \sum_{k=1}^N
\frac{\cos \left(\frac{\pi  k}{L}\left(L+x\right)\right) \sin \left(\frac{2p}{\hbar} (L-\left| x\right| )\right)}{p}.
\end{align*}
We start by considering the case $x=L-\hbar{u}$. Of course $\chi_{[-L,L]}(L-\hbar{u})\to \chi_{\R_+}(u)$ pointwise, as $\hbar\to0$. If we replace  $x=L-\hbar{u}$ with $u\geq0$ in $A$, we obtain a Riemann sum ($\hbar=\mu/N$) of a continuous function, and hence,
$$
A= \frac{\hbar}{2 L}\sum_{k=1}^N\frac{\sin \left(2u\left(\frac{\hbar\pi  k}{2 L}+p\right)\right)}{\frac{\hbar\pi  k}{2 L}+p}\to\frac{1}{2L}\int_0^{\mu}\frac{\sin\left(2u\left(\frac{\pi\xi}{2L}+p\right)\right)}{2u\left(\frac{\pi\xi}{2L}+p\right)}d\xi=\frac{1}{\pi}\int_0^{2u\left(\frac{\pi\mu}{2L}+p\right)}\frac{\sin\xi}{\xi}d\xi.
$$
Similarly, for $B$. We now analyse $C$, for $x=L-\hbar u$,
\begin{align*}
C=\frac{\sin \left(2pu \right)}{p}\frac{\hbar}{2L}\sum_{k=1}^N\cos \left(\frac{\pi  \hbar k}{L}u\right)&\to \frac{\sin \left(2pu\right)}{p}\frac{1}{2L}\int_0^{\mu}\cos \left(\frac{\pi \xi}{L}u\right)d\xi\\
&=\frac{\sin \left(2pu\right)}{p}\frac{1}{2\pi}\sin\left(\frac{\pi\mu}{L}u\right).
\end{align*}

We now consider the case $p=\frac{\pi \mu}{2L}+\frac{\hbar\pi}{2L} {v}$. In this case we get
\begin{align*}
A+B&= \frac{1}{2 L}\sum_{k=1}^N\frac{\sin \left(2(L-|x|)\frac{\pi}{2L}\left(k+N+{v}\right)\right)}{\frac{\pi}{2L}\left(k+N+{v}\right)}+ \frac{1}{2 L}\sum_{k=1}^N\frac{\sin \left(2(L-|x|)\frac{\pi}{2L}\left(k-N-{v}\right)\right)}{\frac{\pi}{2L}\left(k-N-{v}\right)}\\
&= \frac{1}{2 L}\sum_{j={N+1}}^{2N}\frac{\sin \left(2(L-|x|)\frac{\pi}{2L}\left(j+{v}\right)\right)}{\frac{\pi}{2L}\left(j+{v}\right)}+ \frac{1}{2 L}\sum_{j=0}^{N-1}\frac{\sin \left(2(L-|x|)\frac{\pi}{2L}\left(j+{v}\right)\right)}{\frac{\pi}{2L}\left(j+{v}\right)}\\
&= \frac{1}{2 L}\sum_{j={0}}^{2N}\frac{\sin \left(2(L-|x|)\frac{\pi}{2L}\left(j+{v}\right)\right)}{\frac{\pi}{2L}\left(j+{v}\right)}- \frac{1}{2 L}\frac{\sin \left(2(L-|x|)\frac{\pi}{2L}\left(N+{v}\right)\right)}{\frac{\pi}{2L}\left(N+{v}\right)},
\end{align*}
and
\begin{align*}
-2C&= \frac{1}{2 L}\frac{\sin\left(2(L-|x|)\frac{\pi}{2L}\left(N+{v}\right)\right)
}{\frac{\pi}{2L}(N+{v})}\left(-2\sum_{k=1}^N\cos\left(\frac{\pi k}{L}(L+x)\right)\right)\\
&= \frac{1}{2 L}\frac{\sin\left(2(L-|x|)\frac{\pi}{2L}\left(N+{v}\right)\right)
}{\frac{\pi}{2L}(N+{v})}\left(1-\frac{\sin\left((2N+1)\frac{\pi}{4L}(L+x)\right)}{\sin\left(\frac{\pi}{4L}(L+x)\right)})\right).
\end{align*}
Hence,
$$
A+B-2C= \frac{1}{2 L}\sum_{j={0}}^{2N}\frac{\sin \left(2(L-|x|)\frac{\pi}{2L}\left(j+{v}\right)\right)}{\frac{\pi}{2L}\left(j+{v}\right)}+R_N
$$
where
$$R_N=
-\frac{1}{2L}\frac{\sin\left(2(L-|x|)\frac{\pi}{2L}\left(N+{v}\right)\right)
}{\frac{\pi}{2L}(N+{v})}\frac{\sin\left((2N+1)\frac{\pi}{4L}(L+x)\right)}{\sin\left(\frac{\pi}{4L}(L+x)\right)}.
$$
Note that $R_N=0$ for $|x|=L$. When $|x|<L$, we have $|R_N|\to0$ as $N\to\infty$. To see that the series in~\eqref{eq:pointwise2} is convergent, we sandwich the sequence $\left(\frac{\sin \left(2(L-|x|)\frac{\pi}{2L}\left(j+{v}\right)\right)}{\frac{\pi}{2L}\left(j+{v}\right)}\right)_{j\geq0}$ by two alternate sequences that satisfy the Leibniz criterion. The claim is proved.
\end{proof}
For a numerical illustration, see Figure~\ref{fig:symbol}.
\subsection{Semiclassical limit of $ \sigma_{\Pi_N H\Pi_N}^\hbar$}

We now address the problem of genuinely truncated quantum observables $\Pi_N H\Pi_N$, with $\Pi_N$  the spectral projection~\eqref{eq:proj_box} onto the first $N$ levels of a free particle in a box, $\operatorname{Ran}\Pi_N\subset D(H)$ and $[H,\Pi_N]\neq0$. 

We present here two explicit examples that can be thought of as  natural analogues of the truncated momentum operator in the harmonic oscillator eigenbasis of Section~\ref{subs:harmonic}.

\subsubsection{Multiplication by a tridiagonal operator}
Define $H=f(\hat{x})$ as a multiplication operator on $D(f(\hat{x}))=L^2(\R)$ where
\be
\label{eq:special_g}
f(x)=\frac{1}{\sqrt{L}}\sin\left( \frac{\pi}{2L}x\right).
\ee
Note that $f(\hat{x})$ does not commute with the spectral projections $\Pi_N$. We have instead,
\be
\label{eq:tridiag}
f(\hat{x})u_k=-\left(\frac{1}{2\sqrt{L}}u_{k-1}+\frac{1}{2\sqrt{L}}u_{k+1}\right).
\ee
We remark that $f(\hat{x})$ is tridiagonal in the basis $\{u_k\}_{k\geq1}$, and in this sense it is the analogue of the operator $\hat{p}$ for the Hermite wavefunctions considered in Section~\ref{subs:harmonic}. 
From~\eqref{eq:tridiag}  we get an explicit formula for the integral kernel of the truncated observable 
\begin{align}
K_{\Pi_Nf(\hat{x})\Pi_N}(x,y)&=-\frac{1}{2\sqrt{L}}\sum_{k=1}^{N-1}\left[u_k(x)u_{k+1}(y)+u_{k+1}(x)u_{k}(y)\right].
\label{eq:K_N2}
\end{align}
Equivalently, the matrix representation of $\Pi_Nf(\hat{x})\Pi_N$ in the basis $\left\{ u_k \right\}_{k \geq 1}$ is the $N\times N$ real symmetric matrix
\be
\label{eq:matrixm}
-\frac{1}{2\sqrt{L}}
\begin{bmatrix}
0 & 1& 0& 0& \ldots &&\\
1& 0& 1&0& & &\\
0 &1&0 &1&&&\\
0&0&1 &0&& \\ 
\vdots&& && \ddots& \\ 
 & & && & 0&1 \\
 &   &&&&1 &0
\end{bmatrix}.
\ee
We can check the Conditions of Theorem~\ref{thm:main2}:
\begin{enumerate}[label={(C\arabic*)},ref=C\arabic*]
\item $\|\sigma_{\Pi_Nf(\hat{x})\Pi_N}^{\hbar}\|_{L^2}$ is  bounded, and in fact it converges to $\|f\chi_R\|_{L^2}$. Indeed (see the matrix representation~\eqref{eq:matrixm}),
\begin{align*}
\|\sigma_{\Pi_Nf(\hat{x})\Pi_N}^{\hbar}\|^2_{L^2}&=2\pi\hbar \|K_{\Pi_Nf(\hat{x})\Pi_N}^{\hbar}\|^2_{L^2}\\
&=2\pi\hbar\Tr (\Pi_Nf(\hat{x})\Pi_N)^2=
2\pi\hbar \frac{1}{4L}2\sum_{k=1}^{N-1}1^2\to \frac{\pi\mu}{L}.
\end{align*}
On the other hand,
$$
\|f\chi_R\|_{L^2}^2=\int_{R} f(x)^2dxdp=\int_{-\frac{\pi \mu}{2L}}^{\frac{\pi \mu}{2L}}dp\int_{-L}^L \frac{1}{L}\sin^2\left( \frac{\pi}{2L}x\right)dx= \frac{\pi\mu}{L};$$
\item $\sigmaPc$ strongly converges  to $\chi_{R}\in L^2(R_x\times\R_p)$ (see Theorem~\ref{thm:convergenceL2_R});
\item $\lim\limits_{\substack{N\to\infty, \hbar\to0\\\hbar N=\mu}}\|\sigma_{\Pi_N^{\perp}f(\hat{x})\Pi_N}^{\hbar}\|_{L^2}=0$ by the tridiagonal structure of $\Pi_Nf(\hat{x})\Pi_N$:
\begin{align*}
\|\sigma_{\Pi_N^{\perp}f(\hat{x})\Pi_N}\|^2_{L^2}&=\|\sigma_{\Pi_Nf(\hat{x})\Pi_N}^{\hbar}(x,p)-\sigma_{m(\hat{x})\Pi_N}^{\hbar}\|^2_{L^2}\\
&=
2\pi\hbar\|K_{\Pi_Nf(\hat{x})\Pi_N}-K_{f(\hat{x})\Pi_N}\|^2_{L^2}\\
&=2\pi\hbar\frac{1}{4L}\|u_{N+1}\|^2_{L^2}\|u_{N}\|^2_{L^2}=\frac{\pi\hbar}{2L}\to0;
\end{align*}
\end{enumerate}
We conclude that
\be
\lim\limits_{\substack{N\to\infty, \hbar\to0\\\hbar N=\mu}}\|\sigma_{\Pi_N^{}f(\hat{x})\Pi_N}^{\hbar}-f\chi_{R}\|_{L^2(\R_x\times\R_p)}=0.
\ee

\subsubsection{Truncated momentum}

We now consider the operator $H=\hat{p}$ with domain $D(\hat{p})=H^1(\R)$, acting as
$\hat{p}f(x)=-i\hbar f'(x)$.
Note that $\operatorname{Ran}\Pi_N\subset D(\hat{p})$, but $[\Pi_N,\hat{p}]\neq0$. In fact, in contrast to the previous example, the commutator $[\Pi_N,\hat{p}]$ is \emph{not} a finite rank operator. 
The integral kernel of the truncated momentum $\Pi_N\hat{p}\Pi_N$ is
\begin{align}\label{kernelHN}
K_{\Pi_N\hat{p}\Pi_N}(x,y)&=\int_{\R}K_{\Pi_N}(x,z)\left(-i\hbar\frac{\partial}{\partial z}\right)K_{\Pi_N}(z,y)dz=\sum_{j,k=1}^N C^{\hbar}_{j,k}u_j(x)u_k(y),
\end{align}
where
\begin{equation}
C^{\hbar}_{j,k}=\langle u_j|\hat{p}|u_k\rangle=-\frac{i\hbar }{L}\left[1-(-1)^{j+k}\right]\frac{jk}{j^2-k^2}.
\end{equation}

The matrix representation of $\Pi_N\hat{p}\Pi_N$ in the basis $\left\{ u_k \right\}_{k \geq 1}$ is the $N\times N$ complex hermitian matrix
\be
\label{eq:matrixp}
-\frac{i \hbar}{L}
\begin{bmatrix}
0 & -\frac{4}{3}& 0&-\frac{8}{15} & \ldots &\\
\frac{4}{3}& 0& -\frac{12}{5}&0& & \vdots \\
0 &\frac{12}{5} &0 &- \frac{24}{7}&\\ 
\frac{8}{15} &0& \frac{24}{7} &\ddots& \\ \\
\vdots & & & & 0&-\frac{2N(N-1)}{2N-1} \\
 &\cdots   &&&\frac{2N(N-1)}{2N-1} &0
\end{bmatrix}.
\ee

The Weyl symbol of $\Pi_N\hat{p}\Pi_N$ is
\begin{multline}
\label{eq:symbol_HN}
\sigma_{\Pi_N\hat{p}\Pi_N}^{\hbar}(x,p)=\chi_{[-L,L]}(x)\frac{\hbar}{2L}\sum_{1\leq k<j\leq N} C^{\hbar}_{j,k}\times \\
\sum_{\epsilon_1,\epsilon_2=\pm1}
\frac{\epsilon_1\epsilon_2}{\frac{\hbar \pi}{4L}(j+\epsilon_1k)+\epsilon_2p}\sin\left( \frac{\pi}{2L}(j-\epsilon_1k)(x+L)\right) \sin\left(\frac{2}{\hbar}(L-|x|)\left( \frac{\hbar\pi}{4L}(j+\epsilon_1k)+\epsilon_2 p\right) \right).
\end{multline}
To derive~\eqref{eq:symbol_HN} we used~\eqref{eq:symbol_ujuk} and the facts that  $C_{j,k}$ is purely imaginary, $C_{j,k}=\overline{C_{k,j}}$, and $\sigma_{|u_j\rangle\langle u_k|}^{\hbar}=\overline{\sigma_{|u_j\rangle\langle u_k|}^{\hbar}}$.

This operator is the analogue of the truncated momentum of Sec.~\ref{subs:harmonic}. Here we replace the 
 harmonic oscillator spectral projections with the spectral projections of the free particle in a box. This apparently innocuous replacement destroys the tridiagonal structure of the truncated momentum.
 
We now proceed to show that $\sigma_{\Pi_N^{}\hat{p}\Pi_N}^{\hbar}$ does fit in the hypotheses of Theorem~\ref{thm:main2}:
\begin{enumerate}[label={(C\arabic*)},ref=C\arabic*]
\item $\|\sigma_{\Pi_N\hat{p}\Pi_N}^{\hbar}\|_{L^2}$ is  bounded. Indeed,
\begin{align*}
\|\sigma_{\Pi_N\hat{p}\Pi_N}^{\hbar}\|^2_{L^2}=2\pi\hbar \|K_{\Pi_N\hat{p}\Pi_N}^{\hbar}\|^2_{L^2}&=2\pi\hbar\sum_{j,k=1}^N|C_{j,k}^{\hbar}|^2\\
&=\frac{2\pi\hbar^3}{L^2} \sum_{j,k=1}^N\left[1-(-1)^{j+k}\right]^2\left(\frac{jk}{j^2-k^2}\right)^2\\
\end{align*}
We change variables $m=j+k$, $n=j-k$, and bound the sum from above by
\begin{align*}
\frac{2\pi\hbar^3}{L^2} \sum_{\substack{m=2}}^{2N}\sum_{\substack{n=|m-N-1|-(N-1)\\ n\neq0}}^{|m-N-1|+(N-1)}4\left(\frac{m}{n}-\frac{n}{m}\right)^2
&\leq\frac{2\pi\hbar^3}{L^2} \sum_{m=1}^{2N}\sum_{\substack{n=-N\\ n\neq0}}^{N}8\left(\frac{m^2}{n^2}+\frac{n^2}{m^2}\right)\\
&\leq\frac{16\pi\hbar^3}{L^2} 2\left(2\sum_{n=1}^{N}\frac{1}{n^2}\right)\left(\sum_{m=1}^{2N}m^2\right)\\
&\leq \frac{2^9\pi\hbar^3}{L^2}\zeta(2)N^3= \frac{2^9\pi^3\mu^2}{6L^2}.
\end{align*}

\item $\sigmaPc$ converges strongly to $\chi_{R}\in L^2(R_x\times\R_p)$ (see Theorem~\ref{thm:convergenceL2_R});
\item We show that $\lim\limits_{\substack{N\to\infty, \hbar\to0\\\hbar N=\mu}}\|\sigma_{\Pi_N^{\perp}\hat{p}\Pi_N}^{\hbar}\|_{L^2}=0$. Let  $B_{\hbar,N}:=\|\sigma_{\Pi_N^{\perp}\hat{p}\Pi_N}^{\hbar}\|_{L^2}$.  Then, 
\begin{align}
\label{eq:BhN}
B_{\hbar,N}=2\pi\hbar\sum_{\substack{k\leq N\\ j>N}}|C_{j,k}|^2
&\leq\frac{8\pi\hbar^3}{L^2} \sum_{\substack{k\leq N\\ j>N}}\left(\frac{jk}{j^2-k^2}\right)^2=\frac{8\pi\mu^3}{L^2} \frac{1}{N}\frac{1}{N^2}\sum_{\substack{k\leq N\\ j>N}}\left(\frac{jk}{j^2-k^2}\right)^2.
\end{align}
We now write, for $\epsilon=\frac{1}{N}$,
$$
\frac{1}{N^2}\sum_{\substack{k\leq N\\ j>N}}\left(\frac{jk}{j^2-k^2}\right)^2=\frac{1}{N^2}\sum_{\substack{k\leq N-\epsilon N\\ j>N+\epsilon N}}\left(\frac{jk}{j^2-k^2}\right)^2=\int_{0}^{1-\epsilon}{dv}\int_{1+\epsilon}^{\infty}{du}\left(\frac{uv}{u^2-v^2}\right)^2+R_{N}(\epsilon)
$$
for some remainder $R_{N}(\epsilon)$. With the change variables $x=u+v/2$, $y=u-v/2$, the integral simplifies and we can show that
$$
\int_{0}^{1-\epsilon}{dv}\int_{1+\epsilon}^{\infty}{du}\left(\frac{uv}{u^2-v^2}\right)^2\leq C\log\frac{1}{\epsilon}=C\log N,\quad |R_{N}(\epsilon)|\leq C
$$
for some constant $C$ independent of $N$. Plugging these estimates in~\eqref{eq:BhN} we see that $B_{\hbar,N}\to0$.
\end{enumerate}
The latter estimates can be used in fact  to show the stronger condition
\begin{enumerate}[label={(C\arabic*')},ref=C\arabic*']
\item $\|\sigma_{\Pi_N\hat{p}\Pi_N}^{\hbar}\|_{L^2(R_x\times\R_p)}\to \|p\chi_{R}\|_{L^2(R_x\times\R_p)}$. Indeed:
\end{enumerate}
\begin{align*}
\|\sigma_{\Pi_N\hat{p}\Pi_N}^{\hbar}\|^2_{ L^2}&=2\pi\hbar \|K_{\Pi_N\hat{p}\Pi_N}^{\hbar}\|^2_{L^2}\\
&=2\pi\hbar \,\Tr (\Pi_N\hat{p}\Pi_N)^2\\
&=2\pi\hbar \, \Tr (\Pi_N\hat{p}^2\Pi_N)+B_{\hbar,N}\\
&=2\pi\hbar \sum_{k=1}^{N}2E_k+B_{\hbar,N}=\frac{\pi^3\hbar^3}{2L^2} \sum_{k=1}^{N}k^2+B_{\hbar,N}.\end{align*}
where we used the explicit form of $E_k$~\eqref{eq:eigenvalues}. The term $B_{\hbar,N}$ tends to $0$.  
Hence, $\|\sigma_{\Pi_N\hat{p}\Pi_N}^{\hbar}\|^2_{L^2(R_x\times\R_p)}\to \frac{\pi^3\mu^3}{6L^2}$, and also
$$
\|p\chi_R\|^2_{L^2(R_x\times\R_p)}=\int_{R}p^2dxdp=\int_{-\frac{\pi \mu}{2L}}^{\frac{\pi \mu}{2L}}p^2dp\int_{-L}^L dx= \frac{\pi^3\mu^3}{6L^2}.$$
We conclude that there is strong convergence of the truncated momentum
$$
\lim\limits_{\substack{N\to\infty, \hbar\to0\\\hbar N=\mu}}\int_{\R_x\times\R_p}\left|\sigma_{\Pi_N^{}\hat{p}\Pi_N}^{\hbar}(x,p)-p\chi_{R}(x,p)\right|^2dxdp=0.
$$

\section*{Acknowledgements}
Research supported by the Italian National Group of Mathematical Physics (GNFM-INdAM), by PNRR MUR projects CN00000013-`Italian National Centre on HPC, Big Data and Quantum Computing' and PE0000023-NQSTI,  by Regione Puglia through the project `Research for Innovation' - UNIBA024, by Istituto Nazionale di Fisica Nucleare (INFN) through the project `QUANTUM', and by PRIN 2022 project 2022TEB52W-PE1-`The charm of integrability: from nonlinear waves to random matrices'.

\section*{Declarations}
\subsection*{Funding and/or Conflicts of interests/Competing interests} The authors have no competing interests to declare that are relevant to the content of this article.

\subsection*{Data availability} 
Data sharing not applicable to this article as no datasets were generated or analysed during the current study.

\end{document}